\documentclass[11pt]{article}
\usepackage{amsmath}
\usepackage{amsthm}
\usepackage{amssymb}
\usepackage{graphicx} 
\usepackage[margin=1in]{geometry}
\usepackage{changepage}
\usepackage{color}
\usepackage{enumerate}
\usepackage{natbib}
\usepackage[T1]{fontenc}

\usepackage{mathrsfs}
\usepackage{mathtools}
\usepackage{braket}
\usepackage{stmaryrd}
\usepackage{xfrac}

\usepackage{wrapfig}
\usepackage{floatrow}

\usepackage{caption}
\usepackage{subcaption}
\usepackage[dvipsnames]{xcolor}
\usepackage{tikz}
\usetikzlibrary{
    positioning,
    shapes.geometric,
    shapes.callouts,
    shapes.misc,
    patterns,
    intersections,
    arrows.meta
}
\tikzset{>=latex}
\newcommand{\plotSimplex}{
    \draw [draw=none, fill=verylightgray] (0,0)--(0.5,0.866)--(1,0)--(0,0);
    
    \draw[lightgray!80, thin] (0.75,0.433)--(0,0);  
    \draw[lightgray!80, thin] (0.25,0.433)--(1,0);  
    \draw[lightgray!80, thin] (0.5,0)--(0.5,0.866);  

    \draw [thick,black!50] (0,0)--(0.5,0.866)--(1,0)--(0,0);  
        
    \node[anchor=north] at (0,0) {$(1,0,0)$};
    \node[anchor=north] at (1,0) {$(0,1,0)$};
    \node[anchor=south] at (0.5,0.866) {$(0,0,1)$};
}
\def\ax{0}
\def\ay{0}
\def\bx{0.5}
\def\by{0.866}
\def\cx{1}
\def\cy{0}
\newcommand{\plotSimplexWithNonExtremeBox}{
    \draw [draw=none, fill=verylightgray] (0,0)--(0.5,0.866)--(1,0)--(0,0);

    \draw [draw=none, fill=black!14] (\ax+1/2*\bx-1/2*\ax,\ay+1/2*\by-1/2*\ay)--(\ax+1/2*\bx-1/2*\ax,\ay+1/2*\by-1/2*\ay)--(\bx+1/2*\cx-1/2*\bx,\by+1/2*\cy-1/2*\by)--(\bx+1/2*\cx-1/2*\bx,\by+1/2*\cy-1/2*\by)--(\cx+1/2*\ax-1/2*\cx,\cy+1/2*\ay-1/2*\cy)--(\cx+1/2*\ax-1/2*\cx,\cy+1/2*\ay-1/2*\cy);

    \draw[lightgray!80, thin] (0.75,0.433)--(0,0);  
    \draw[lightgray!80, thin] (0.25,0.433)--(1,0);  
    \draw[lightgray!80, thin] (0.5,0)--(0.5,0.866);  

    \draw [thick,black!50] (0,0)--(0.5,0.866)--(1,0)--(0,0);  
        
    \node[anchor=north] at (0,0) {$(1,0,0)$};
    \node[anchor=north] at (1,0) {$(0,1,0)$};
    \node[anchor=south] at (0.5,0.866) {$(0,0,1)$};
}
\tikzstyle{aggregate}=[draw, thick, minimum size=8,inner sep=0,regular polygon,regular polygon sides=3]
\tikzstyle{voter}=[draw, thick, minimum size=7,inner sep=0,circle]
\tikzstyle{label}=[text opacity=1]

\definecolor{color1}{HTML}{d97904}
\colorlet{color2}{blue!50!black}
\definecolor{verylightgray}{HTML}{f1f1f1}

\DeclareMathOperator*{\argmax}{arg\,max}
\DeclareMathOperator*{\argmin}{arg\,min}
\DeclareMathOperator*{\med}{med}

\newcommand{\F}{\mathcal{F}}
\newcommand{\V}{\mathcal{V}}
\newcommand{\G}{\mathcal{G}}
\newcommand{\A}{\mathcal{A}}
\newcommand{\Amean}{\mu}
\newcommand{\N}{\mathbb{N}}
\newcommand{\multivoter}[2]{{#2 \times #1}}
\newcommand{\p}{p}   
\newcommand{\symvote}[1]{{\textstyle p(#1)}}
\newcommand{\ellone}[2]{\lVert#1 - #2\rVert_1}
\newcommand{\ellinf}[2]{\lVert#1 - #2\rVert_\infty}
\newcommand{\I}{\mathcal{I}}

\newcommand{\eps}{\varepsilon}

\makeatletter
    \newtheorem*{rep@theorem}{\rep@title}
        \newcommand{\newreptheorem}[1]{%
        \newenvironment{reptheorem}[1]{%
            \def\rep@title{\Cref{##1}}%
            \begin{rep@theorem}%
        }{\end{rep@theorem}}%
     }
    \newreptheorem
\makeatother

\renewcommand\paragraph{\@startsection{paragraph}{4}{\z@}%
                                    {1.25ex \@plus1ex \@minus.2ex}%
                                    {-1em}%
                                    {\normalfont\normalsize\bfseries}}

\usepackage{enumitem}
\usepackage[colorlinks=true, allcolors=blue!50!black]{hyperref}
\usepackage[capitalize, nameinlink]{cleveref}
\usepackage{etoolbox}

\newtheorem{theorem}{Theorem}
\newtheorem{maintheorem}{Theorem}
\newtheorem{lemma}{Lemma}

\newtheorem{corollary}{Corollary}
\newtheorem{proposition}{Proposition}
\newtheorem{observation}{Observation}

\usepackage{thmtools}

\declaretheoremstyle[%
  headfont=\normalfont\itshape,%
  qed=\qedsymbol%
]{proofsketchstyle} 
\declaretheorem[name={Proof sketch},style=proofsketchstyle,unnumbered]{proof_sketch}

\declaretheoremstyle[%
  headfont=\normalfont\itshape,%
  qed=$\blacksquare$,%
  postheadspace = -2.0em, 
]{proofofclaimstyle} 
\declaretheorem[name={Proof of claim},style=proofofclaimstyle,unnumbered,postheadhook=\begin{adjustwidth}{2.5em}{0pt},prefoothook=\end{adjustwidth}]{proof_of_claim}

\usepackage{thm-restate}

\Crefname{lemmaenumi}{Lemma}{Lemmas}
\AtBeginEnvironment{lemma}{%
    \crefalias{enumi}{lemmaenumi}%
    \setlist[enumerate,1]{
        label={\textit{(\roman*)}},
        ref={\thelemma(\roman*)}
    }%
}

\newcommand{\appendColon}[1]{\textbf{#1:}}
\setlistdepth{9}
\newlist{CaseTree}{enumerate}{9}
\setlist[CaseTree,1]{label=Case \arabic*, leftmargin = 19px, itemindent=8mm, font=\appendColon}
\setlist[CaseTree,2]{label*=.\arabic*,leftmargin = *, itemindent=10mm, font=\appendColon}
\setlist[CaseTree,3]{label*=.\arabic*,leftmargin = *, itemindent=12mm, font=\appendColon}
\setlist[CaseTree,4]{label*=.\arabic*,leftmargin = *, itemindent=14mm, font=\appendColon}
\setlist[CaseTree,5]{label*=.\arabic*,leftmargin = *, itemindent=16mm, font=\appendColon}
\setlist[CaseTree,6]{label*=.\arabic*,leftmargin = *, itemindent=18mm, font=\appendColon}
\setlist[CaseTree,7]{label*=.\arabic*,leftmargin = *, itemindent=20mm, font=\appendColon}
\setlist[CaseTree,8]{label*=.\arabic*,leftmargin = *, itemindent=22mm, font=\appendColon}
\setlist[CaseTree,9]{label*=.\arabic*,leftmargin = *, itemindent=24mm, font=\appendColon}

\title{Truthful Budget Aggregation:\\Beyond Moving-Phantom Mechanisms }

\date{}

\usepackage{authblk}

\author[1]{Mark de Berg}
\author[2]{Rupert Freeman}
\author[1]{Ulrike Schmidt-Kraepelin}
\author[1]{Markus Utke}
\affil[1]{TU Eindhoven, The Netherlands}
\affil[2]{Darden School of Business, University of Virginia, VA, USA}
\begin{document}

\maketitle
\vspace*{-1.5cm}

\begin{abstract}
    We study a budget-aggregation setting in which a number of voters report their ideal distribution of a budget over a set of alternatives, and a mechanism aggregates these reports into an allocation. Ideally, such mechanisms are truthful, i.e., voters should not be incentivized to misreport their preferences. 
    For the case of two alternatives, the set of mechanisms that are truthful and additionally meet a range of basic desiderata (anonymity, neutrality, and continuity) exactly coincides with the so-called moving-phantom mechanisms, but whether this space is richer for more alternatives was repeatedly stated as an open question. We answer this question in the affirmative by presenting a class of truthful mechanisms that are not moving-phantoms but satisfy the three properties.
    Since moving-phantom mechanisms can only provide limited fairness guarantees (measured as the worst-case distance to a fair share solution), one motivation for broadening the class of truthful mechanisms is the hope for improved fairness guarantees. We dispel this hope by showing that lower bounds holding for the class of moving-phantom mechanisms extend to all truthful, anonymous, neutral, and continuous mechanisms. 
\end{abstract}

\section{Introduction}
Consider a scenario where a perfectly divisible resource, such as money or time, needs to be distributed among various alternatives while taking the preferences of a group of voters into account. This task, known as \emph{portioning}, lies at the heart of participatory budgeting, a voting model gaining increasing attention thanks to its pivotal role in civic participation initiatives \citep{aziz2021participatory,cabannes2004participatory}. We study a variant called \emph{budget aggregation} \citep{lindner2008midpoint, goel2019knapsack, freeman2021truthful}, where voters express their favorite allocation over alternatives (also called \emph{projects}), and their dissatisfaction with an outcome is measured by its $\ell_1$-distance from their ideal distribution. 
Unlike many traditional voting scenarios, budget aggregation with $\ell_1$-utilities opens the door to mechanisms that incentivize truthfulness among voters. In fact, \cite{freeman2021truthful} present a whole class of \emph{truthful} mechanisms called \emph{moving-phantom mechanisms}. Loosely speaking, moving-phantom mechanisms are an extension of (the neutral subclass of) \emph{generalized median rules} \citep{moulin1980strategy}, characterized as the only truthful mechanisms in the two-alternative setting meeting the criteria of anonymity and continuity \citep{moulin1980strategy,masso2011strategy}, to elections with more than two alternatives. However, the question of whether moving-phantom mechanisms are the \emph{only} truthful, continuous, anonymous, and neutral\footnote{
Informally, a mechanism is \textit{truthful} if a voter can never decrease the $\ell_1$-distance between its favorite allocation and the mechanism's outcome by reporting an allocation that is not its favorite.
A mechanism is \textit{anonymous} (\textit{neutral}, respectively) if the outcome does not depend on the identity of the voters (alternatives, respectively), and it is \textit{continuous} if it is continuous according to the standard definition, when interpreted as a function.} mechanisms in the general case has remained open and was repeatedly mentioned in recent literature \citep{freeman2021truthful, caragiannis2022truthful, freeman2024project,brandt2024optimal}.
We resolve this question.

\begin{reptheorem}{thm:existence_truthful_non_phantoms}[informal]
    There exists a budget-aggregation mechanism that is truthful, anonymous, neutral, and continuous but not a moving-phantom mechanism.
\end{reptheorem}

To prove \Cref{thm:existence_truthful_non_phantoms}, we define the class of \emph{cutoff-phantom} mechanisms by combining moving-phantom mechanisms with a cutoff function that redistributes budget away from any project that the moving-phantom mechanism assigns more than a certain threshold share of the budget. Cutoff-phantoms are well defined for any moving-phantom mechanism and any threshold at least $\sfrac{1}{2}$, but in general both components need to be chosen carefully to preserve truthfulness. We identify one novel moving-phantom mechanism, \textsc{GreedyMax}, for which all of the corresponding cutoff-phantoms (one per choice of threshold) are truthful.

While cutoff-phantoms significantly expand the class of known (anonymous, neutral, continuous) truthful mechanisms, they fail to satisfy unanimity, which prescribes that, whenever the voters all agree on their most preferred distribution, the mechanism should output that distribution. What happens when we add unanimity to our list of properties? In Section~\ref{sec:truthful_unanimous_non_phantoms}, we provide a partial answer to this question, by giving a mechanism that is not a moving-phantom mechanism but is unanimous and truthful (as well as anonymous, neutral, and continuous) for instances with two voters and three alternatives. While primarily a proof of concept, this result suggests that the class of truthful budget-aggregation mechanisms satisfying other desirable properties might not allow for a concise description.

One motivation to search for alternative truthful mechanisms stems from the desire for mechanisms that are not only truthful but also fair. To this end, one might consider the \textsc{mean} mechanism---the mechanism that averages the voters' reports for every alternative---as a benchmark. This mechanism appears to be particularly appealing since it is equivalent to assigning each of the $n$~voters their equal \textit{fair share} of the budget and letting them allocate this budget according to their ideal distribution. However, doing so might not be in the voter's best interest: it is well-known and intuitive that mean aggregation incentivizes voters to extremize their reported preference in order to bring the mean closer to their true preference. In other words, the mean mechanism violates truthfulness.
A natural question, then, is how much fairness needs to be compromised in order to restore truthfulness.

To quantify (violations of) fairness, \cite{caragiannis2022truthful} proposed measuring the worst-case deviation, in terms of $\ell_1$-distance, of a mechanism's outcome from the mean. They introduced the \textsc{PiecewiseUniform} mechanism, ensuring an $\ell_1$-distance of $\sfrac{2}{3} + \varepsilon$ for some constant~$\varepsilon < 10^{-5}$, for the case of three alternatives.
\cite{freeman2024project} presented the \textsc{Ladder} mechanism, establishing an upper bound of $\sfrac{2}{3}$ for three alternatives and non-trivial bounds for up to six alternatives. They demonstrate that the ladder mechanism results in a worst-case $\ell_{\infty}$-distance from the mean of $\frac{m-1}{2m}$, where $m$ is the number of alternatives. While a lower bound provided by \cite{caragiannis2022truthful} indicates the tightness of these results within the class of moving-phantom mechanisms, the best lower bounds for $\ell_1$-approximation (and $\ell_{\infty}$-approximation, respectively) for three alternatives within the class of all truthful mechanisms stand at $\sfrac{1}{2}$ (and $\sfrac{1}{4}$, respectively).

In this paper, we show lower bounds that match the best known lower bounds for the class of moving-phantom mechanisms. These bounds are known to be tight for $\ell_\infty$-approximation and for $\ell_1$-approximation with $m=3$, but a gap remains for $\ell_1$ with $m>3$.

\begin{reptheorem}{thm:truthful_uniform}[informal] Let $m$ be the number of alternatives. There exists a budget-aggregation instance for which every truthful, anonymous, neutral, and continuous mechanism returns an outcome with $\ell_{\infty}$-distance of $\frac{m-1}{2m}$ and $\ell_1$-distance of $\frac{m-1}{m}$ from the mean.
\end{reptheorem}

The remainder of the paper is organized as follows. After discussing related work and introducing the problem formally,  \Cref{sec:existence_truthful_non_phantom} is devoted to proving \Cref{thm:existence_truthful_non_phantoms}. We first introduce the \textsc{GreedyMax} moving-phantom mechanism and show that it can be combined with any constant threshold greater than $\sfrac{1}{2}$ to yield a truthful mechanism that is not a moving-phantom. We then define the class of \emph{slow} moving-phantom mechanisms,
and show that every mechanism in this class can be combined with \emph{some} threshold function to yield a truthful, non-moving-phantom mechanism. We construct the unanimous and truthful non-phantom mechanism for $n=2$ and $m=3$ in \Cref{sec:truthful_unanimous_non_phantoms}, and prove \Cref{thm:truthful_uniform} in \Cref{sec:lowerBound}. 

\paragraph{Related Work.}
Our work contributes to a growing literature on budget aggregation. \citet{lindner2008midpoint} and \citet{goel2019knapsack} study the rule that maximizes utilitarian welfare, the neutral version of which turns out to be the unique Pareto-efficient moving-phantom mechanism~\citep{freeman2021truthful}. \citet{caragiannis2022truthful} introduce the paradigm of mean approximation for the budget aggregation problem, which is built upon by~\citet{freeman2024project}. \citet{brandt2024optimal} show that no mechanism can be truthful, Pareto-efficient, and proportional,\footnote{Proportionality~\citep{freeman2021truthful} says that, on instances where every voter prefers to spend the entire budget on a single alternative, the mechanism should output the mean of the votes.} generalizing a result of \citet{freeman2021truthful} that held only for moving-phantom mechanisms. \citet{elkind2023settling} axiomatically study several budget-aggregation mechanisms, and find that the mean performs well relative to the other rules they consider. \citet{goyal2023low} work in a similar setting to ours  (except that every alternative's funding is capped by its predefined cost) and study mechanisms with low sample complexity in terms of their distortion.

For the special case of two alternatives, it is known that truthful and anonymous budget-aggregation mechanisms are characterized by generalized median rules~\citep{moulin1980strategy,masso2011strategy}. Generalized median rules are parameterized by $n+1$ ``phantom'' votes, with the output being the median of these phantom votes and the $n$ submitted votes. Several papers have used generalized median mechanisms to truthfully approximate the mean in the two-alternative setting~\citep{renault2005protecting,renault2011assessing,caragiannis2016truthful,jennings2023new}, with the optimal approximation stated explicitly by~\cite{caragiannis2022truthful}. However, for higher numbers of alternatives,~\cite{caragiannis2022truthful} obtained a lower bound for general truthful mechanisms that diverged from their lower bound for moving-phantom mechanisms. Our \Cref{thm:truthful_uniform} closes this gap.

Beyond budget aggregation, portioning has been studied with other input models including ordinal preferences~\citep{airiau2023portioning}, dichotomous preferences~\citep{bogomolnaia2005collective,brandl2021distribution,michorzewski2020price}, or more general cardinal utility functions over alternatives~\citep{fain2016core,wagner2023strategy}. We refer the reader to the survey of~\citet{aziz2021participatory} for additional discussion of the participatory budgeting literature.

Finally, we remark that a weaker version of our \Cref{thm:existence_truthful_non_phantoms} was recently independently obtained by \citet[Appendix D]{brandt2024optimal}. Specifically, the authors show that for the case of a single voter, there exists a truthful, anonymous, neutral, and continuous mechanism that is not a moving-phantom mechanism. Such a mechanism can then be extended to a mechanism satisfying the four properties by concatenating it with any moving-phantom mechanism for the cases of more voters. However, this seems unsatisfactory, since such a mechanism is only not a moving-phantom mechanism for the case of one voter. Our result is significantly stronger since we design mechanisms that are not moving-phantom mechanisms for any number of voters and alternatives.

\section{Preliminaries} \label{sec:prelim}

For $n \in \N$, we let $[n] = \{k \in \N \mid 1 \le k \le n\}$. For $k \in \N$, we write $\Delta^{(k)} = \{ v \in [0,1]^{k+1} \mid \lVert v \rVert_1 = 1 \}$ for the standard $k$-simplex. 
In the setting of budget aggregation we have $n$ \textit{voters} deciding how to distribute a budget of~1 over $m \ge 2$ \textit{alternatives}.
Each voter $i \in [n]$ reports an allocation $\p_i \in \Delta^{(m-1)}$ as their \textit{vote} and all votes together form the \textit{profile} $P = (\p_1, \dots, \p_n)$. We denote by $\mathcal{P}_{n,m} = (\Delta^{(m-1)})^n$, the set of all profiles with $n$ voters and $m$ alternatives.

\paragraph{Budget-Aggregation Mechanisms.} A \emph{budget-aggregation mechanism} $\A$ (sometimes shortened to \emph{mechanism}) is a family of functions $\A_{n,m} : \mathcal{P}_{n,m} \rightarrow \Delta^{(m-1)}$, one for every pair $n,m \in N, m \ge 2$, that map each profile $P \in \mathcal{P}_{n,m}$ to an \emph{aggregate} $a = \A_{n,m}(P) \in \Delta^{(m-1)}$. Since $n$ and $m$ are always clear from context, we slightly abuse notation and write $\mathcal{A}$ instead of $\mathcal{A}_{n,m}$.

A natural budget-aggregation mechanism is the \textsc{Mean} mechanism $\Amean(p_1, \dots, p_n) = \frac{1}{n} \sum_{i \in [n]} \p_i.$ 
As discussed earlier, the \textsc{Mean} mechanism is appealing from a fairness standpoint but violates \emph{truthfulness}, which we formalize below.\footnote{Consider two voters with ideal allocations $(1,0)$ and $(0.5,0.5)$. The \textsc{mean} mechanism would allocate $(0.75,0.25)$, but the second voter can enforce their ideal allocation by reporting $(0,1)$. \label{foot:not-truthful}}  For measuring the satisfaction of voters with a certain aggregate, we assume that they evaluate the aggregate with respect to the $\ell_1$-distance to their ideal budget allocation, i.e., the disutility of voter~$i$ for aggregate $a$ is given by $\ellone{\p_i}{a}$.

\paragraph{Truthfulness.} A budget-aggregation mechanism $\A$ is \emph{truthful} if for any $n, m \in \N$ with $m \ge 2$ and any profile $P = (\p_1, \dots, \p_n) \in \mathcal{P}_{n,m}$, voter $i \in [n]$, and misreport $\p_i^\star \in \Delta^{(m-1)}$, the following holds for profile $P^\star = (\p_1, \dots, \p_{i-1}, \p_i^\star, \p_{i+1}, \dots, \p_n)$: $\ellone{\p_i}{\A(P)} \le \ellone{\p_i}{\A(P^\star)}.$

In contrast to many other voting settings, budget aggregation under $\ell_1$-disutilities permits truthful mechanisms that are not dictatorships. In fact, \cite{freeman2021truthful} introduce a whole class of truthful mechanisms, which we define next.

\paragraph{Moving-Phantom Mechanisms.}
A moving-phantom mechanism is defined via so-called phantom systems, one for each number of voters. For $n \in \N$,
a \textit{phantom system} $\F_n = \{f_k \mid k=0,\dots,n\}$ is a set of $(n+1)$ non-decreasing, continuous functions $f_k: [0,1] \rightarrow [0,1]$ with $f_k(0) = 0$, $f_k(1) = 1$ and $f_0(t) \ge \dots \ge f_n(t)$ for all $t \in [0,1]$.
Then, a family of phantom systems $\F = \{\F_n \mid n \in \N\}$ defines the moving-phantom mechanism $\A^\F$: 
For a given profile $P = (\p_1, \dots, \p_n)$ with $p_i = (p_{i,1}, \dots, p_{i,m})$ for any $i \in [n]$, 
let $t^\star \in [0,1]$ be a value for which
$$
    \sum_{j \in [m]} \med(f_0(t^\star), \dots, f_n(t^\star), \p_{1,j}, \dots, \p_{n,j}) = 1, 
$$
where $\med(\cdot)$ denotes the median.\footnote{Here, the median is well-defined since the input always consists of an odd number of values.} The moving-phantom mechanism $\A^\F$ returns the allocation $a=(a_1,\ldots,a_m)$ with
$$
    a_j = \med(f_0(t^\star), \dots, f_n(t^\star), \p_{1,j}, \dots, \p_{n,j})
$$
for each $j \in [m]$. While $t^*$ is not always unique, the induced allocation is and we refer to $t^*$ as a \emph{time of normalization}. \cite{freeman2021truthful} show that any moving-phantom mechanism is truthful. We also refer to moving-phantom mechanisms as \textit{phantom mechanisms}.

\paragraph{Anonymity, Neutrality, and Continuity.} If we would ask whether moving-phantom mechanisms are the only truthful mechanisms, the answer would clearly be no. In particular, there exist truthful mechanisms that treat voters or alternatives highly unequally, or that drastically change their output in response to small changes in the input. Informally, the overarching goal of \Cref{thm:existence_truthful_non_phantoms} is to understand whether moving-phantom mechanisms are the only ``reasonable'' truthful budget-aggregation mechanisms. To make this statement formal, we introduce the properties of anonymity, neutrality, and continuity, which we deem as non-negotiable.

A budget-aggregation mechanism $\A$ is \emph{anonymous} if, for any profile $(\p_1, \dots, \p_n)$ and any permutation $\sigma: [n] \rightarrow [n]$, the following holds:
    $
        \A(\p_1, \dots, \p_n) = \A(\p_{\sigma(1)}, \dots, \p_{\sigma(n)}).
    $

A budget-aggregation mechanism $\A$ is \emph{neutral} if, for any profile $P = (\p_1, \dots, \p_n)$ with $\p_i = (\p_{i,1}, \dots, \p_{i,m})$ for $i \in [n]$, and any permutation $\sigma: [m] \rightarrow [m]$, the following holds. Let $P^\sigma = (\p_1^\sigma, \dots, \p_n^\sigma)$, where $\p_i^\sigma = (\p_{i,\sigma(1)}, \dots, \p_{i,\sigma(m)})$ for $i \in [n]$, and let $\A(P) = (a_1, \dots, a_m)$. Then
    $
        \A(P^\sigma) = (a_{\sigma(1)}, \dots, a_{\sigma(m)}).
    $

We say that a mechanism $\A$ is \emph{continuous} if, for any $n,m \in \N$, the function $\A_{n,m}$ is continuous with respect to the standard definition of a continuous function. 

\cite{freeman2021truthful} show that the class of moving-phantom mechanisms not only satisfies truthfulness, but also meets all of the other criteria introduced above. Moreover any mechanism satisfying the four properties coincides with a moving-phantom mechanism for all two-alternative profiles.

\begin{theorem}[\cite{freeman2021truthful}] \label{thm:phantoms_anonymous_neutral_continuous}\label{thm:characterization_m2}
    Any moving-phantom mechanism is truthful, anonymous, neutral, and continuous. For any mechanism $\A$ that is truthful, anonymous, neutral, and continuous, there exists a moving-phantom mechanism $\A^\F$ such that $\A(P) = \A^{\F}(P)$ for all $n \in \N$ and $P \in \mathcal{P}_{n,2}$.
\end{theorem}

In \Cref{sec:existence_truthful_non_phantom}, we show that the second part of \Cref{thm:characterization_m2} does not extend to more alternatives. 
All missing proofs can be found in the appendix.

\section{Truthful Non-Phantom Mechanisms}
\label{sec:existence_truthful_non_phantom}

We now formalize and prove \Cref{thm:existence_truthful_non_phantoms}, which states that there exist budget-aggregation mechanisms that are truthful, anonymous, neutral, and continuous, but not moving-phantom mechanisms.

\begin{maintheorem}[formal]
\label{thm:existence_truthful_non_phantoms}
    There exists a budget-aggregation mechanism $\A$ that is truthful, anonymous, neutral, and continuous, but for no $n,m \in \N$ with $m \ge 3$, does there exist a phantom system $\F_n$ with $\A(P)=\A^{\F_n}(P)$ for all $P \in \mathcal{P}_{n,m}$.
\end{maintheorem}

\begin{wrapfigure}{R}{0.5\textwidth}
    \centering
    \vspace*{6mm}
    \begin{tikzpicture}[scale=4.5]
        \plotSimplexWithNonExtremeBox
        \small

        \draw[thick, draw=black]
            (\ax+1/2*\bx-1/2*\ax,\ay+1/2*\by-1/2*\ay) -- (\cx+1/2*\ax-1/2*\cx,\cy+1/2*\ay-1/2*\cy);
        \draw[thick, draw=black]
            (\bx+1/2*\cx-1/2*\bx,\by+1/2*\cy-1/2*\by) -- (\ax+1/2*\bx-1/2*\ax,\ay+1/2*\by-1/2*\ay);
        \draw[thick, draw=black]
            (\cx+1/2*\ax-1/2*\cx,\cy+1/2*\ay-1/2*\cy) -- (\bx+1/2*\cx-1/2*\bx,\by+1/2*\cy-1/2*\by);

        \node (vote1) [draw, circle, color1, inner sep=0, minimum size=5, fill] at (0.2,0.2) {};
        \node (vote1cut) [draw, cross out, very thick, color2, inner sep=0, minimum size=3, fill] at (0.337,0.283) {};
        \draw [dashed, ->, thin, shorten <= 1pt, shorten >= 2pt] (vote1) -- (vote1cut);

        \node (vote2) [draw, circle, color1, inner sep=0, minimum size=5, fill] at (0.48,0.83) {};
        \node (vote2cut) [draw, cross out, very thick, color2, inner sep=0, minimum size=3, fill] at (0.48,0.432) {};
        \draw [dashed, ->, thin, shorten <= 1pt, shorten >= 2pt] (vote2) -- (vote2cut);

        \node (vote3) [draw, circle, color1, inner sep=0, minimum size=5, fill] at (0.6,0.3) {};
        \node (vote3cut) [draw, cross out, very thick, color2, inner sep=0, minimum size=3, fill] at (0.6,0.3) {};
    \end{tikzpicture}
    
    \caption{Illustration of $c_\tau$ for $m = 3$ and $\tau = \frac{1}{2}$. Circles are example inputs and crosses are the corresponding outputs. Solid black lines correspond to points where one coordinate is exactly $\tau$. %
    \vspace*{-8mm}
    }
    \label{fig:cutoff}
\end{wrapfigure}
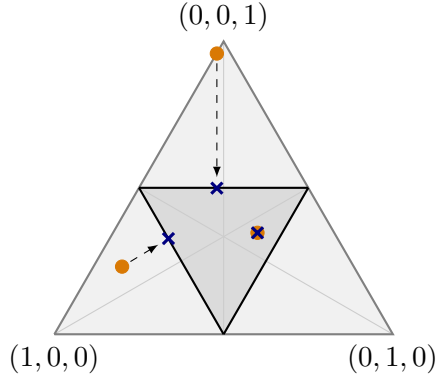
To prove \Cref{thm:existence_truthful_non_phantoms}, we introduce a class of budget-aggregation mechanisms, all of which are continuous, anonymous, and neutral, and some of which are also truthful and not moving-phantom mechanisms. We call this class of mechanisms \textit{cutoff-phantoms}. The key idea is to apply a so-called \textit{cutoff} function to the output of a phantom mechanism in situations where some alternative receives more than a certain threshold. After introducing cutoff-phantoms, we introduce a new moving-phantom mechanism called \textsc{GreedyMax}, which behaves particularly well in combination with cutoff functions. More precisely, we show in \Cref{subsec:greedymax} that \textsc{CutoffGreedyMax} is not a phantom mechanism and satisfies truthfulness for any constant threshold above $\sfrac{1}{2}$. In \Cref{sec:moretruthful}, we present more examples of cutoff-phantoms that are truthful, albeit for non-constant thresholds.

We modify the outcome of a phantom mechanism in cases where the aggregate is ``extreme''. More precisely, for some \emph{threshold function} $\tau: \mathbb{N} \times \mathbb{N} \rightarrow [\sfrac{1}{2},1]$ we say that $a\in \Delta^{(m-1)}$ is $(j,\tau)$\emph{-extreme} if $a_j > \tau(n,m)$ for $j\in[m]$. While $\tau$ will be constant in this section, in \Cref{sec:moretruthful} it can depend on $m$ and $n$. We also write $j$-extreme instead of $(j,\tau)$-extreme whenever $\tau$ is clear from the context.

If a phantom mechanism returns a $(j,\tau)$-extreme outcome, we lower the value of alternative $j$ to $\tau$ and redistribute the removed budget evenly among all other alternatives. We refer to this function on the aggregate as the \textit{cutoff function}. More formally, let $m \in \N$, $m \ge 2$, and $a = (a_1, \dots, a_m) \in \Delta^{(m-1)}$. For some $\tau: \mathbb{N} \times \mathbb{N} \rightarrow [\sfrac{1}{2},1]$ we define the cutoff function $c_\tau: \Delta^{(m-1)} \rightarrow \Delta^{(m-1)}$ as follows. Let $k \in \argmax_{j \in [m]}(a_j)$, then 
\[
    c_\tau(a)_j = \begin{cases} 
    a_j & \text{if } a_k \le \tau, \\
    \tau(n,m) & \text{if } a_k > \tau \text{ and } j = k, \\
    a_j + \frac{a_k-\tau(n,m)}{m-1} & \text{if } a_k > \tau \text{ and } j \neq k. \\
    \end{cases}
\]
Note, that $c_\tau(a)\neq a$ only if $\max_{j \in [m]}(a_j) > \tau(n,m) \ge \sfrac{1}{2}$, which in turn implies that $k$ is unique and $c_\tau$ is well defined. 
A visualization of the cutoff function can be found in \Cref{fig:cutoff}.

\paragraph{Cutoff-phantom Mechanisms.}  
For any family of phantom systems $\mathcal{F}$ and $\tau: \N \times \N \rightarrow [\frac{1}{2},1]$, 
we define the \textit{cutoff-phantom} mechanism $\A^\F_{c_\tau}$ as the composition of $\A^\F$ and $c_{\tau}$, i.e., $$\A^\F_{c_{\tau}} = c_{\tau} \circ \A^\F.$$ We drop $\tau$ from the subscript of $c$ whenever it is clear from context. 

\begin{restatable}{proposition}{cutoffCAN}\label{prop:greedy_max_cutoff_anonymous_neutral_continuous}
    Any cutoff-phantom mechanism is anonymous, neutral, and continuous.
\end{restatable}

\subsection{\textsc{CutoffGreedyMax}} \label{subsec:greedymax}

We introduce a novel moving-phantom mechanism called \textsc{GreedyMax} and then show that \textsc{CutoffGreedyMax} is truthful and not a phantom mechanism for any constant threshold $\tau \in [\frac{1}{2},1)$. Among known moving-phantom mechanisms in the literature, \textsc{GreedyMax} is the only one that induces a cutoff-phantom that is truthful for constant thresholds (see Section~\ref{sec:moretruthful}).
\textsc{GreedyMax} moves the first $n$ phantoms from zero to one simultaneously and leaves the last phantom at zero.\footnote{According to the definition of a moving-phantom mechanism, each phantom $f_i$ is required to have $f_i(1) = 1$. For this phantom system, however, normalization is always reached when the top $n$ phantoms reach $1$ (or sooner). Thus, we fix the last phantom at position zero for simplicity.}

\vspace{-0.2cm}
\paragraph{\textsc{GreedyMax}.} We define the moving-phantom mechanism \textsc{GreedyMax} as the mechanism $\A^\G$ induced by the family of phantom systems $\G = \{\G_n \mid n \in \N\}$ with
\[
    \G_n = \{g_0, \dots, g_{n}\} ,\,\,\text{ where } g_0(t) = \dots = g_{n-1}(t) = t \,\,\text{ and }\,\, g_n(t) = 0.
\]

The medians according to the \textsc{GreedyMax} phantom systems follow a simple pattern:

\begin{observation} \label{obs:medians}
    Let $t \in [0,1]$, and $P \in \mathcal{P}_{n,m}$. Then for all $j \in [m]$, $$\med(g_0(t), \dots, g_n(t), \p_{1,j}, \dots, \p_{n,j}) =\min(t, \max_{i \in [n]} p_{i,j}).$$
\end{observation}

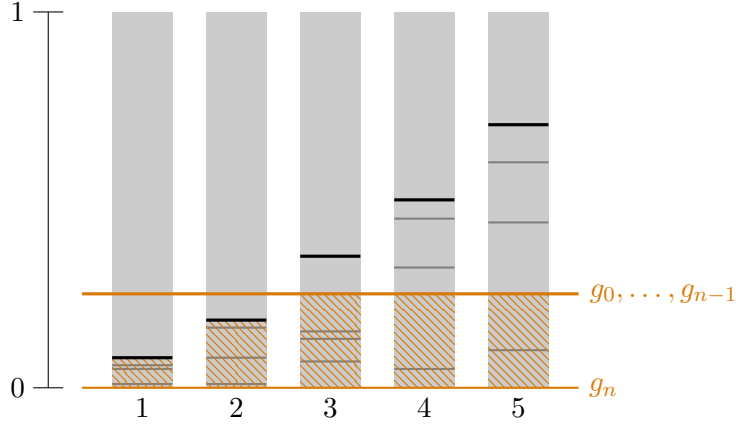
\begin{figure}
    \centering
    \begin{tikzpicture}[scale=5]
        \def\coorddist{0.25}
        \def\coordwidth{0.08}
        
        \draw (0,0) -- (0,1);  
        \draw (-.04,0) -- (0.04,0); 
        \draw (-.04,1) -- (0.04,1); 
        \node[anchor=east, xshift=-5px] at (0,0) {$0$}; 
        \node[anchor=east, xshift=-5px] at (0,1) {$1$}; 
        \filldraw[fill=black!20,draw=none] (1*\coorddist-\coordwidth,0) rectangle (1*\coorddist+\coordwidth,1); 
        \filldraw[fill=black!20,draw=none] (2*\coorddist-\coordwidth,0) rectangle (2*\coorddist+\coordwidth,1); 
        \filldraw[fill=black!20,draw=none] (3*\coorddist-\coordwidth,0) rectangle (3*\coorddist+\coordwidth,1); 
        \filldraw[fill=black!20,draw=none] (4*\coorddist-\coordwidth,0) rectangle (4*\coorddist+\coordwidth,1); 
        \filldraw[fill=black!20,draw=none] (5*\coorddist-\coordwidth,0) rectangle (5*\coorddist+\coordwidth,1); 
        
        \node[anchor=north] at (1*\coorddist,0) {$1$};
        \node[anchor=north] at (2*\coorddist,0) {$2$};
        \node[anchor=north] at (3*\coorddist,0) {$3$};
        \node[anchor=north] at (4*\coorddist,0) {$4$};
        \node[anchor=north] at (5*\coorddist,0) {$5$};

        \draw[thick, gray] (1*\coorddist-\coordwidth,0.06) -- (1*\coorddist+\coordwidth,0.06); 
        \draw[thick, gray] (1*\coorddist-\coordwidth,0.05) -- (1*\coorddist+\coordwidth,0.05); 
        \draw[thick, gray] (1*\coorddist-\coordwidth,0.01) -- (1*\coorddist+\coordwidth,0.01); 
        \fill [color1, pattern=north west lines, pattern color=color1] (1*\coorddist-\coordwidth,0.0) rectangle (1*\coorddist+\coordwidth,0.08);
        \draw[very thick] (1*\coorddist-\coordwidth,0.08) -- (1*\coorddist+\coordwidth,0.08); 
        
        \draw[thick, gray] (2*\coorddist-\coordwidth,0.16) -- (2*\coorddist+\coordwidth,0.16); 
        \draw[thick, gray] (2*\coorddist-\coordwidth,0.08) -- (2*\coorddist+\coordwidth,0.08); 
        \draw[thick, gray] (2*\coorddist-\coordwidth,0.01) -- (2*\coorddist+\coordwidth,0.01); 
        \fill [color1, pattern=north west lines, pattern color=color1] (2*\coorddist-\coordwidth,0.0) rectangle (2*\coorddist+\coordwidth,0.18);
        \draw[very thick] (2*\coorddist-\coordwidth,0.18) -- (2*\coorddist+\coordwidth,0.18); 
        
        \draw[thick, gray] (3*\coorddist-\coordwidth,0.15) -- (3*\coorddist+\coordwidth,0.15); 
        \draw[thick, gray] (3*\coorddist-\coordwidth,0.13) -- (3*\coorddist+\coordwidth,0.13); 
        \draw[thick, gray] (3*\coorddist-\coordwidth,0.07) -- (3*\coorddist+\coordwidth,0.07);
        \fill [color1, pattern=north west lines, pattern color=color1] (3*\coorddist-\coordwidth,0.0) rectangle (3*\coorddist+\coordwidth,0.25);
        \draw[very thick] (3*\coorddist-\coordwidth,0.35) -- (3*\coorddist+\coordwidth,0.35); 
        
        \draw[thick, gray] (4*\coorddist-\coordwidth,0.45) -- (4*\coorddist+\coordwidth,0.45); 
        \draw[thick, gray] (4*\coorddist-\coordwidth,0.32) -- (4*\coorddist+\coordwidth,0.32); 
        \draw[thick, gray] (4*\coorddist-\coordwidth,0.05) -- (4*\coorddist+\coordwidth,0.05);
        \fill [color1, pattern=north west lines, pattern color=color1] (4*\coorddist-\coordwidth,0.0) rectangle (4*\coorddist+\coordwidth,0.25);
        \draw[very thick] (4*\coorddist-\coordwidth,0.50) -- (4*\coorddist+\coordwidth,0.50); 
                
        \draw[thick, gray] (5*\coorddist-\coordwidth,0.60) -- (5*\coorddist+\coordwidth,0.60); 
        \draw[thick, gray] (5*\coorddist-\coordwidth,0.44) -- (5*\coorddist+\coordwidth,0.44); 
        \draw[thick, gray] (5*\coorddist-\coordwidth,0.10) -- (5*\coorddist+\coordwidth,0.10);
        \fill [color1, pattern=north west lines, pattern color=color1] (5*\coorddist-\coordwidth,0.0) rectangle (5*\coorddist+\coordwidth,0.25);
        \draw[very thick] (5*\coorddist-\coordwidth,0.70) -- (5*\coorddist+\coordwidth,0.70); 

        \draw[thick,color1] (1*\coorddist-2*\coordwidth,0) -- (5*\coorddist+2*\coordwidth,0) node[anchor=west]{$g_n$}; 
        \draw[very thick,color1] (1*\coorddist-2*\coordwidth,0.25) -- (5*\coorddist+2*\coordwidth,0.25) node[anchor=west]{$g_0, \dots, g_{n-1}$};

  \end{tikzpicture}
    \caption{Illustrative example of the \textsc{GreedyMax} mechanism. The votes on each alternative are marked by (gray) lines. Since the mechanism only depends on the maximum vote on each alternative, we emphasize the maximum votes and order the alternatives in increasing order of their maxima. 
    The phantom positions are shown as (orange) lines crossing all alternatives and the upper bounds of the hatched areas mark the allocations to each of the alternatives.}
    \label{fig:greedy_max_phantom}
\end{figure}

This observation implies that $\sfrac{1}{m} \le t^\star \le 1$. At time $t$, the sum of all medians is upper bounded by $mt$, and therefore $t^\star \ge \sfrac{1}{m}$. Additionally, when $t=1$ the median on each alternative is exactly the maximum vote on that alternative, and their sum is at least 1, which implies that $t^\star \le 1$.

Observation~\ref{obs:medians}, which we also illustrate in \Cref{fig:greedy_max_phantom}, leads to an alternative interpretation of the mechanism that also inspired its name \textsc{GreedyMax}: In the beginning, we set a budget $b=1$ and a counter $k=0$. Then, we iterate over the alternatives in increasing order of their maximum votes. Let $j$ be the current alternative. If $\max_{i \in [n]} p_{i,j} < b/(m-k)$, then assign $\max_{i \in [n]} p_{i,j}$ to alternative $j$, decrease $b$ by the same value and increase $k$ by one. Otherwise, assign $b/(m-k)$ to all remaining alternatives and stop the process. The interpretation is that the mechanism greedily assigns the value of the maximum vote to alternatives, unless this would lead to future alternatives receiving less than what is assigned in the current step.

To complete the proof of \Cref{thm:existence_truthful_non_phantoms}, we show that \textsc{CutoffGreedyMax} is not a moving-phantom mechanism and that it is truthful in \Cref{prop:greedy_max_cutoff_not_phantom} and \Cref{prop:greedy_max_cutoff_truthfulness}, respectively. 

\begin{proposition} \label{prop:greedy_max_cutoff_not_phantom}
    For every $n,m \in \N$, $m \ge 3$, every phantom system $\F_n$ and $\tau \in [\frac{1}{2}, 1)$, there exists a profile $P \in \mathcal{P}_{n,m}$ such that $\A^\G_{c_\tau}(P) \neq \A^{\F_n}(P)$.
\end{proposition}
\begin{proof}
    We fix $n,m \in \N$, $m \ge 3$ and $\tau \in [\frac{1}{2}, 1)$ and consider the profile $P \in \mathcal{P}_{n,m}$ defined as $$P = \Big(\Big(1, 0, \dots,0\Big), \dots, \Big(1, 0, \dots,0\Big),\Big(\frac{1+\tau}{2}, \frac{1-\tau}{2}, 0, \dots, 0\Big)\Big).$$ \textsc{GreedyMax} reaches normalization at $t = \frac{1+\tau}{2}$ and $\A^\G(P) = (\frac{1+\tau}{2}, \frac{1-\tau}{2}, 0, \dots, 0)$. Thus $$a = \A^\G_c(P) = \Big(\tau, \frac{1-\tau}{2}+\frac{1-\tau}{2(m-1)}, \frac{1-\tau}{2(m-1)}, \dots, \frac{1-\tau}{2(m-1)}\Big).$$ Suppose for contradiction that there exists a moving-phantom mechanism $\A^\F$ such that $\A^\F(P) = \A^\G_c(P)$. Then, $a_j$ corresponds to the median of the votes and the phantoms for alternative $j$ for each $j\in [m]$. Since $a_3$ does not match any of the votes for alternative $3$, 
    some phantom must be at position $\frac{1-\tau}{2(m-1)}$ at a time of normalization.
    But then, for alternative $2$ there are $n-1$ votes at $0$, one vote at $\frac{1-\tau}{2}$ and at least one phantom at $\frac{1-\tau}{2(m-1)} < \frac{1-\tau}{2}$. Thus the median for alternative $2$ is at most $\frac{1-\tau}{2} < \frac{1-\tau}{2}+\frac{1-\tau}{2(m-1)} = a_2$, a contradiction.
\end{proof}

 We now turn to sketching that \textsc{CutoffGreedyMax} is in fact truthful for any constant $\tau \in [\sfrac{1}{2},1)$. 
 The proof distinguishes two cases based on whether \textsc{GreedyMax} returns an extreme aggregate for the truthful report of a voter. To support this case distinction, \Cref{lem:greedy_max_high_aggregate} characterizes the profiles for which \textsc{GreedyMax} returns extreme outcomes.

\begin{restatable}{lemma}{lemmatwo}\label{lem:greedy_max_high_aggregate}
    For all $n,m \in \N, m \ge 2$, $\tau \in [\frac{1}{2}, 1)$, any profile $P = (\p_1, \dots, \p_n) \in \mathcal{P}_{n,m}$ and corresponding aggregate of \textsc{GreedyMax} $a = \A^\G(P)$ it holds that $a_j > \tau$ if and only if
    $$\min_{i \in [n]}(\p_{i,j}) > \tau \qquad \textnormal{and} \qquad \sum_{\substack{k \in [m] \\ k\neq j}} \max_{i \in [n]}(\p_{i,k}) < 1 - \tau.$$
\end{restatable}

Making use of \Cref{lem:greedy_max_high_aggregate}, we provide a proof sketch of the truthfulness of \textsc{CutoffGreedyMax}.

\begin{restatable}{proposition}{maxcutoffgreedy}\label{prop:greedy_max_cutoff_truthfulness}
    \textsc{CutoffGreedyMax} is truthful for any $\tau \in [\frac{1}{2},1)$. 
\end{restatable}
\begin{proof_sketch}
Consider profiles $P=(p_1,\dots,p_{n-1},p_n)$ and  $P^\star = (p_1,\dots,p_{n-1},p^\star_n)$ where voter $n$ changes their report unilaterally. Let $a$ and $a^*$ be the aggregates of \textsc{GreedyMax} for $P$ and $P^\star$, respectively. Similarly, let $\bar{a}$ and $\bar{a}^\star$ be the aggregates of \textsc{CutoffGreedyMax} for $P$ and $P^\star$, respectively. Then, we distinguish two cases: 
(i) If $a$ is not extreme, then $\bar{a}=a$ and we can show that $\bar{a}^\star$ can be achieved through a misreport by voter $n$ under \textsc{GreedyMax}. This follows directly if $\bar{a}^\star = a^\star$. Otherwise, we prove that \textsc{GreedyMax} returns $\bar{a}^\star$ when voter $v$ reports $\bar{a}^\star$. Hence, $\ellone{v}{\bar{a}^\star} \geq \ellone{v}{\bar{a}}$ follows directly from the truthfulness of \textsc{GreedyMax}. (ii) The aggregate $a$ is extreme, and wlog $m$-extreme. By \Cref{lem:greedy_max_high_aggregate}, all reports in $P$ are $m$-extreme and $\sum_{j \in [m-1]} \max_{i \in [n]}\{p_{i,j}\} \leq 1-\tau$. From the former statement, we know that $\bar{a}_m = \tau \leq p_{n,m}$. Using the latter, we can moreover show that $\bar{a}_j \geq a_j = \max_{i \in [n]} (p_{i,j})\ge \p_{n,j}$ for all $j \in [m-1]$. In words, $\bar{a}$ gives more budget than voter $n$ would like to all alternatives besides $m$. Here, we can show that assuming a truthfulness violation of \textsc{CutoffGreedyMax}, i.e.\ $\ellone{v}{\bar{a}} > \ellone{v}{\bar{a}^\star}$, implies $\bar{a}^\star_m > \bar{a}_m = \tau$, a contradiction.
\end{proof_sketch}

\Cref{thm:existence_truthful_non_phantoms} now follows immediately from \Cref{prop:greedy_max_cutoff_anonymous_neutral_continuous,prop:greedy_max_cutoff_truthfulness,prop:greedy_max_cutoff_not_phantom}.

\subsection{More Truthful Cutoff-Phantoms} \label{sec:moretruthful}

Given the generality of the class of cutoff-phantoms, it is natural to ask whether other moving-phantom mechanisms also retain their truthfulness when combined with a cutoff function. The answer to this question is subtle: The moving-phantom mechanisms that were introduced in the literature fail truthfulness for any constant threshold, however, if we allow the threshold to depend on $m$ and $n$, truthfulness for some of them can be restored. More precisely, the phantom mechanisms \textsc{MaxUtilitarianWelfare} \citep{freeman2021truthful}, \textsc{IndependentMarkets (IM)} \citep{freeman2021truthful}, \textsc{PiecewiseUniform} \citep{caragiannis2022truthful}, and \textsc{Ladder} \citep{freeman2024project}, all fail truthfulness when combined with a constant threshold. We defer their definitions to \Cref{sec:appthreetwo}. 

\begin{restatable}{proposition}{noConstantThreshold}
    For any constant threshold $\tau \in [\frac{1}{2},1)$, \textsc{CutoffMaxUtilitarianWelfare}, \textsc{CutoffIM}, \textsc{CutoffPiecewise}, and \textsc{CutoffLadder} violate truthfulness. 
    \label{prop:noconstantthreshold}
\end{restatable}

We can regain the truthfulness of certain cutoff-phantoms that are induced by a subclass of phantom systems: We call a  family of phantom systems $\F$ \emph{slow}, if, for any $n \in \N$ there exists some point in time at which the second to last phantom has already moved while the first phantom has not reached $1$ and at the point in time where the first phantom reaches $1$, the last phantom is still at $0$. Formally, there exists $t \in [0,1]$ such that $f_0(t)<1$, $f_{n-1}(t) >0$ and for the minimum value of $t$ with $f_0(t) = 1$ we have $f_n(t)=0$. Among the four moving-phantom mechanisms discussed above, \textsc{IM} and \textsc{Ladder} are slow while \textsc{MaxUtilitarianWelfare} and \textsc{Piecewise} are not slow.

\begin{restatable}{theorem}{propCutOffPhantomsTruthful} \label{prop:phantom_cutoff_truthfulness}
   For any slow family of phantom systems $\F$ there exists $\tau: \N \times \N \rightarrow [\frac{1}{2},1)$ such that the cutoff-phantom $\A^{\F}_{c_{\tau}}$ is truthful and not a moving-phantom mechanisms, i.e., for any $n,m \in \N$, $m \ge 3$, every phantom system $\I_n$, there exists a profile $P \in \mathcal{P}_{n,m}$ such that $\A^\F_{c_\tau}(P) \neq \A^{\I_n}(P)$. 
\end{restatable}

Our proof is constructive, i.e., for any slow phantom system we provide a threshold function $\tau$ for which the statement holds. Namely, for $n,m \in \N$, let $t \in [0,1]$ be maximum such that $f_0(t) \le 1-(m-1) \cdot f_{n-1}(t)$. Then, define $\tau(n,m) = 1 -f_{n-1}(t)$. This results in the following threshold functions for \textsc{CutoffIM} and \textsc{CutoffLadder}, respectively: $$\tau(n,m) = \frac{n+m-2}{n+m-1} \quad \quad \text{and} \quad \quad \tau(n,m) = \frac{mn-1}{mn}.$$ 
Proposition~\ref{prop:noconstantthreshold} does not rule out smaller non-constant threshold functions for \textsc{CutoffIM} and \textsc{CutoffLadder}, but the instance used in the proof does preclude threshold functions that do not approach 1 as $n \to \infty$.
The proof technique for the first part of \Cref{prop:phantom_cutoff_truthfulness} (truthfulness) is similar to the proof sketch presented in \Cref{subsec:greedymax}. Namely, we first show that an analogous version of \Cref{lem:greedy_max_high_aggregate} holds for slow moving-phantom mechanisms and the threshold functions introduced above. Additionally, we show a second lemma which characterizes the form of extreme aggregates. Having established these two lemmas, we can then apply the same proof as for \Cref{prop:greedy_max_cutoff_truthfulness}. For the second part of the proof (not a moving-phantom mechanism) we can use the same profile as in the proof of \Cref{prop:greedy_max_cutoff_not_phantom}.

\section{Truthful and Unanimous Non-Phantom Mechanisms}
\label{sec:truthful_unanimous_non_phantoms}

\Cref{prop:phantom_cutoff_truthfulness} defines a broad subclass of cutoff-phantoms that are truthful and not moving-phantoms. However, the definition of cutoff-phantoms is arguably unnatural in that they are defined by ruling out extreme aggregates, and this is the case even if every voter votes for the same extreme distribution. That is, cutoff-phantoms violate unanimity.

\paragraph{Unanimity.} A budget-aggregation mechanism $\A$ is \emph{unanimous} if, for any profile $P = (p, \ldots, p)$, it holds that $\mathcal{A}(P)=p$.

It is natural to ask whether there exists a truthful, anonymous, neutral, and continuous mechanism that is additionally unanimous. We do not have a complete answer to this question, but provide a partial result in support of a positive answer.

\begin{restatable}{theorem}{unanimousproof}\label{thm:existence_unanimous_truthful_non_phantoms}
    There exists a budget-aggregation mechanism $\A$ that is truthful, unanimous, anonymous, neutral, and continuous for $n=2$ and $m=3$, and there does not exist a phantom system $\F_2$ with $\A(P)=\A^{\F_2}(P)$ for all $P \in \mathcal{P}_{2,3}$.
\end{restatable}

To prove \Cref{thm:existence_unanimous_truthful_non_phantoms}, we first introduce a second approach for generating truthful non-phantoms. Namely, instead of ``cutting off'' the aggregate of a moving-phantom mechanism, we suggest to cut off the voter reports before applying a moving-phantom mechanism. To obtain a truthful mechanism via this approach, we introduce a novel moving-phantom mechanism called \textsc{GreedyMin} (as the name suggests, it can be seen as a dual of \textsc{GreedyMax}). We then show in \Cref{app:votecutgreedy} that \textsc{VoteCutGreedyMin} (defined by applying a constant cutoff function to the voter reports and then running \textsc{GreedyMin}), is truthful for any $n$ and $m=3$. 

Clearly, this approach is not sufficient to additionally guarantee unanimity. Hence, in \Cref{app:unamvotecutgreedy}, we define a new cutoff function for $n=2$ that is applied to an extreme vote if and only if the other vote is not extreme. In particular, if both votes are identical, then the cutoff is not applied, preserving unanimity. As an additional subtlety, to guarantee continuity and truthfulness, we only partially apply the cutoff to an extreme vote if the other vote is close to extreme. While our mechanism is carefully designed with concrete thresholds,
even proving truthfulness for the specific mechanism and the restricted case of $n=2$, $m=3$ is nontrivial - the proof relies on a lengthy case analysis inspired by a visual exploration of the rule's behavior. 
While we are optimistic that our proof can be extended to show truthfulness for larger $n$, our arguments do not naturally translate to higher values of $m$.

\section{Lower Bound for Mean-Approximation} \label{sec:lowerBound}

As argued in \Cref{sec:prelim}, the \textsc{mean} mechanism is not truthful. However, due to its intuitive appeal and practical relevance, \cite{caragiannis2022truthful} define the concept of a $\delta$-approximation to the mean as a fairness measure of a mechanism.

\textbf{Mean approximation.} For $\delta: \N \times \N \rightarrow \mathbb{R}$, a mechanism $\A$ achieves a \textit{$\delta$-approximation} to the mean if for any $m,n \in \N$, $m\ge 2$ and profile $P \in \mathcal{P}_{n,m}$ $$\ellone{\A(P)}{\Amean(P)} \le \delta(m,n).$$
\cite{freeman2024project} define a similar notion using the $\ell_\infty$-distance instead of the $\ell_1$-distance. This can be interpreted as a measure that targets fairness towards alternatives, since the maximum deviation to the mean over all alternatives is considered. Consequently, they call their notion project fairness.\footnote{\cite{freeman2024project} refer to alternatives as \emph{projects}.}

\textbf{Project fairness.} For $\delta: \N \times \N \rightarrow \mathbb{R}$, a mechanism $\A$ is \textit{$\delta$-project fair} if for all $m,n \in \N$, $m\ge 2$ and profiles $P \in \mathcal{P}_{n,m}$ $$\ellinf{\A(P)}{\Amean(P)} \le \delta(m,n).$$ 

\citet[Theorem 7]{caragiannis2022truthful} show that no moving-phantom mechanism can achieve a $\delta$-approximation with $\delta(m,n) < \frac{m-1}{m}$ for all $m \in \N$ and even $n \in \N$ by providing a specific family of profiles. \citet[Proposition 2]{freeman2024project} use the same profiles to show that no moving-phantom mechanism can achieve $\delta$-project fairness with $\delta(m,n) < \frac{m-1}{2m}$ for all $m\in \N$ and even $n \in N$, and with $\delta(m,n) < \frac{m-1}{2m}\cdot\frac{n-1}{n}$ for all $m\in \N$ and odd $n \in N$. However, the arguments used do not easily extend outside of the class of moving-phantom mechanisms. Instead, \cite{caragiannis2022truthful} use a two-voter profile to show that no truthful mechanism can achieve a $\delta$-approximation to the mean with $\delta(m,2) < \frac{1}{2}$ for any $m \in \N$. One can verify that the same arguments show that no truthful mechanism can achieve $\delta$-project fairness with $\delta(m,2) < \frac{1}{4}$ for any $m \in \N$. In this section, we directly strengthen the results of \cite{caragiannis2022truthful} and \cite{freeman2024project} by showing that the $\delta$-approximation and $\delta$-project fairness lower bounds for moving-phantom mechanisms in fact apply to any continuous, truthful, anonymous and neutral mechanism. To do so we consider the same profiles but characterize not only the output of moving-phantom mechanisms but of all mechanisms that fall into this class, which requires a significantly more detailed argument. In the following, we refer to mechanisms satisfying the four properties of \textbf{c}ontinuity, \textbf{t}ruthfulness, \textbf{a}nonymity, and \textbf{n}eutrality as \emph{CTAN} mechanisms.

\begin{figure}[ht]
        \begin{subfigure}[t]{0.475\textwidth}
            \begin{tikzpicture}[scale=5.5]
                \plotSimplex
                \small
                \node[draw, circle, minimum size=20, fill=white, fill opacity=1] at (0.5,0.286) (callout){};
                \draw[color1, fill] (0.5,0.286) node[voter, fill, xshift=-2.5, yshift=3.75](center1){};
                \draw[color2] (0.5,0.286) node[voter, fill=white, xshift=2.5, yshift=-3.75](center2){};
                \node[color2, label, anchor=north west] at (center2) {$\symvote{\frac{1}{3}}$};    
                \node[color1, label, anchor=south east, xshift=1pt] at (center1) {$\symvote{\frac{1}{3}}$};    
                \draw[color1, fill] (0,0) node[voter, fill]{} node[label, anchor=south east]{$\symvote{1}$};
                \draw[color2] (0.75,0.433) node[voter, fill=white]{} node[label, anchor=south west]{$\symvote{0}$};    
            
                \draw[color1] (0.35*0.75,0.35*0.433) node[aggregate, fill]{} node[label, anchor=south east, xshift=-0.5pt]{$\A(P) = \symvote{\beta}$};    
                \draw[color2] (0.85*0.75,0.85*0.433) node[aggregate, fill=white]{} node[label, anchor=north west, xshift=-1pt]{$\A(P') = \symvote{\gamma}$};    
                 
            \end{tikzpicture}
            \caption{Profiles $P$ (orange, filled) and $P'$ (blue, outline) with their corresponding aggregates. 
            }
            \label{fig:lower_bound_n2_two_profiles}
        \end{subfigure}%
        \hfill
        \begin{subfigure}[t]{0.475\textwidth}
            \begin{tikzpicture}[scale=5.5]
                \plotSimplex
                \small
                \node[draw, circle, minimum size=20, fill=white, fill opacity=1] at (0.5,0.286) (callout){};
                
                \draw[color1, fill] (callout) node[voter, fill, xshift=-4.1, yshift=2.9](center1){};
                \draw[color2] (callout) node[voter, fill=white, xshift=4.1, yshift=2.9](center2){};
                \node[color2, label, anchor=south west, xshift=-1pt] at (center2) {$\symvote{\frac{1}{m}}$};   
                \node[color1, label, anchor=south east, yshift=0pt, xshift=2pt] at (center1) {$\symvote{\frac{1}{m}}$};    
                \draw[color1, fill] (0.25*0.75,0.25*0.433) node[voter, fill]{} node[label, anchor=north east, xshift=-1pt, yshift=2pt]{$\symvote{\alpha}$};    
                \draw[color2] (0.5,0) node[voter, fill=white]{} node[label, anchor=south west]{$\hat{p}$};
    
                \draw[color1] (0.47*0.75,0.47*0.433) node[aggregate, fill]{} node[label, anchor=east, yshift=1.8pt, xshift=-2pt]{$\A(P^\star) = \symvote{\beta'}$};  
                \draw[color2] (callout) node[aggregate, fill=white, xshift=0, yshift=-5](center3){} node[label, anchor=north west, xshift=-2pt, yshift=-4pt]{$\A(\hat{P}) = \symvote{\frac{1}{m}}$};  
    
                \draw[dashed, ->, thin, shorten <= 5pt, shorten >= 5pt] (0.5,0) --(0.25*0.75,0.25*0.433);
            \end{tikzpicture}
            \caption{Profiles $P^\star$ (orange, filled) and $\hat{P}$ (blue, outline) with their corresponding aggregates.}
            \label{fig:lower_bound_n2_parametric_profile}
        \end{subfigure}
        \caption{
        Proof sketch of \Cref{thm:truthful_uniform} for $m = 3$ and $n = 2$. In the figure, circles denote voters and triangles denote aggregates. 
        For a value $\alpha \in [0,1]$ we write the vote $(\alpha, \frac{1-\alpha}{2}, \frac{1-\alpha}{2})$ as $\symvote{\alpha}$.
        All voters and aggregates in the white circle in the center are positioned at $p(\frac{1}{3}) = (\frac{1}{3}, \frac{1}{3}, \frac{1}{3})$. 
        The goal of the proof is to show that for the profile $P$ (orange profile in \Cref{fig:lower_bound_n2_two_profiles}), the aggregate of any CTAN mechanism has to be in the center, i.e., $\beta=\frac{1}{3}$.
        To do so, we first show that for all profiles in \Cref{fig:proof-sketch-main-theorem-two}, the aggregate must lie on the line segment connecting their two voter reports. Then, we suppose for contradiction that in profile $P$ it holds that $\beta > \frac{1}{3}$ and consider profile $P'$ (blue profile, \Cref{fig:lower_bound_n2_two_profiles}) with aggregate $p(\gamma)$, with $\gamma\leq\frac{1}{3}$. If $\gamma < \frac{1}{3}$, we can then modify $P$ and $P'$ in such a way that their aggregates do not change but that the modified profiles are identical -- a contradiction. If $\gamma=\frac{1}{3}$ then we consider the profile $\hat{P}$ (blue profile, \Cref{fig:lower_bound_n2_parametric_profile}) along with $P^\star$ (orange profile, \Cref{fig:lower_bound_n2_parametric_profile}), a modified version of $P$ which is chosen so that its aggregate $p(\beta')$ is sufficiently close to the center. Accordingly, the voter at $\hat{p}$ prefers the orange aggregate $p(\beta')$ over the blue aggregate $p(\frac{1}{m})$ and therefore can manipulate by switching to $p(\alpha)$, effectively turning profile $\hat{P}$ into profile $P^\star$ and obtaining a preferred outcome, a contradiction to truthfulness.
        }
        \label{fig:proof-sketch-main-theorem-two}
    \end{figure}
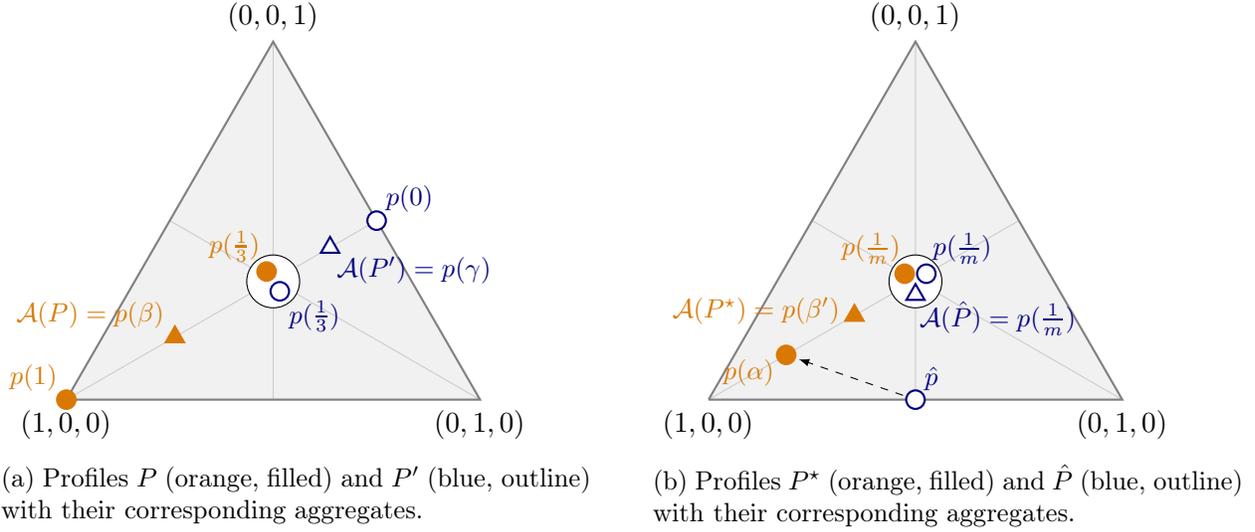

\begin{restatable}[formal]{maintheorem}{maintheoremtwo} \label{thm:truthful_uniform}
For any  $m, n \in \N$, $m \ge 2$ there exists a profile in $\mathcal{P}_{n,m}$ for which every truthful, anonymous, neutral, and continuous mechanism returns an outcome with  
    \begin{itemize}
        \item $\ell_1$-distance of $\frac{m-1}{m}$and $\ell_{\infty}$-distance of $\frac{m-1}{2m}$ to the mean if $n$ is even and
        \item $\ell_1$-distance of $\frac{m-1}{m}\cdot\frac{n-1}{n}$ and $\ell_{\infty}$-distance of $\frac{m-1}{2m}\cdot\frac{n-1}{n}$ to the mean if $n$ is odd.
    \end{itemize} 
\end{restatable}

The profile that we use to show \Cref{thm:truthful_uniform} is 
$P = (\p_1, \dots, \p_n)$ with 
\[
    \p_i = 
    \begin{cases}
        (\frac{1}{m}, \dots, \frac{1}{m}) & \text{ if } i \le \lceil\frac{n}{2}\rceil,\\    
        (1, 0, \dots, 0) & \text{ if } i > \lceil\frac{n}{2}\rceil.  
    \end{cases}
\]
The proof makes use of the symmetry of this and similar profiles and shows that the outcome of any CTAN mechanism for profile $P$ is $(\frac{1}{m},\dots,\frac{1}{m})$.
We provide a sketch of the proof of \Cref{thm:truthful_uniform} in \Cref{fig:proof-sketch-main-theorem-two} (and its caption).

For completeness, we state improved versions of the theorems from \cite{caragiannis2016truthful} and \cite{freeman2024project} that immediately follow from \cref{thm:truthful_uniform}.

\begin{corollary}
    For any $\delta: \N \times \N \rightarrow \N$, such that there are $m,n \in \N$ with $\delta(m,n) < \frac{m-1}{m}$ if $n$ is even or $\delta(m,n) < \frac{m-1}{m}\cdot\frac{n-1}{n}$ if $n$ is odd, there is no truthful, anonymous, neutral, and continuous mechanism that achieves a $\delta$-approximation to the mean.
\end{corollary}

This corollary strengthens the theorem from \cite{caragiannis2022truthful} which gives the same bound of $\frac{m-1}{m}$ for even $n$, but only for the class of moving-phantom mechanisms.

\begin{corollary}
    For any $\delta: \N \times \N \rightarrow \N$, such that there are $m,n \in \N$ with $\delta(m,n) < \frac{m-1}{2m}$ if $n$ is even or $\delta(m,n) < \frac{m-1}{2m}\cdot\frac{n-1}{n}$ if $n$ is odd, there is no truthful, anonymous, neutral, and continuous mechanism that achieves $\delta$-project fairness.
\end{corollary}

This corollary mirrors the theorem from \cite{freeman2024project}, which states the same bounds only for the class of moving-phantom mechanisms. It has further implications, namely, \cite{freeman2024project} additionally show that the \textsc{Ladder} moving-phantom mechanism that they introduce achieves $\frac{m-1}{2m}$-project fairness, which they call the \textit{essentially} optimal project fairness guarantee among all moving-phantom mechanisms, where \textit{essentially} refers to the gap of $O(\frac{1}{n})$ for odd $n$. Our result thus implies that the \textsc{Ladder} mechanism achieves essentially optimal project fairness among the more general class of CTAN mechanisms.

\section{Conclusion}
We resolved in the affirmative the question of whether there exist continuous, truthful, anonymous, and neutral (CTAN) mechanisms that are not moving-phantom mechanisms. 
To do so, we designed a new family of mechanisms by carefully pairing (new) moving-phantom mechanisms with cutoff functions.

Our work raises many interesting questions. Characterizing the space of CTAN mechanisms remains a compelling open question, but there is no natural formulation to conjecture in light of our cutoff-phantom mechanisms. A related question is whether moving-phantom mechanisms can be characterized by adding additional axioms. Our result in \Cref{sec:truthful_unanimous_non_phantoms} shows that adding the natural property of \emph{unanimity} won't be enough. More generally, do there exist non-moving-phantom mechanisms that are competitive with moving-phantom mechanisms with respect to their axiomatic properties and/or fairness guarantees? From the perspective of truthful mean approximation, we have established that the \textsc{Ladder} moving-phantom mechanism provides the optimal worst-case approximation in terms of $\ell_\infty$-distance (for any $m$) and $\ell_1$-distance (for $m=3$). However, for $\ell_1$-distance and $m > 3$ it remains possible that moving-phantom mechanisms do not achieve the optimal approximation among all truthful mechanisms. 

\pagebreak
\paragraph{Acknowledgments.} Part of this research was carried out while Ulrike Schmidt-Kraepelin was supported by the National Science Foundation under Grant No. DMS-1928930 and by the Alfred P. Sloan Foundation under grant G-2021-16778 while being in residence at the Simons Laufer Mathematical Sciences Institute (formerly MSRI) in Berkeley, California, during the Fall 2023 semester.

We also thank Felix Brandt, Matthias Greger, and Warut Suksompong for interesting discussions on the topic.

\bibliography{literature}

\begin{thebibliography}{23}
\providecommand{\natexlab}[1]{#1}
\providecommand{\url}[1]{\texttt{#1}}
\expandafter\ifx\csname urlstyle\endcsname\relax
  \providecommand{\doi}[1]{doi: #1}\else
  \providecommand{\doi}{doi: \begingroup \urlstyle{rm}\Url}\fi

\bibitem[Airiau et~al.(2023)Airiau, Aziz, Caragiannis, Kruger, Lang, and Peters]{airiau2023portioning}
S.~Airiau, H.~Aziz, I.~Caragiannis, J.~Kruger, J.~Lang, and D.~Peters.
\newblock Portioning using ordinal preferences: Fairness and efficiency.
\newblock \emph{Artificial Intelligence}, 314:\penalty0 103809, 2023.

\bibitem[Aziz and Shah(2021)]{aziz2021participatory}
H.~Aziz and N.~Shah.
\newblock Participatory budgeting: Models and approaches.
\newblock \emph{Pathways Between Social Science and Computational Social Science: Theories, Methods, and Interpretations}, pages 215--236, 2021.

\bibitem[Bogomolnaia et~al.(2005)Bogomolnaia, Moulin, and Stong]{bogomolnaia2005collective}
A.~Bogomolnaia, H.~Moulin, and R.~Stong.
\newblock Collective choice under dichotomous preferences.
\newblock \emph{Journal of Economic Theory}, 122\penalty0 (2):\penalty0 165--184, 2005.

\bibitem[Brandl et~al.(2021)Brandl, Brandt, Peters, and Stricker]{brandl2021distribution}
F.~Brandl, F.~Brandt, D.~Peters, and C.~Stricker.
\newblock Distribution rules under dichotomous preferences: Two out of three ain't bad.
\newblock In \emph{Proceedings of the 22nd ACM Conference on Economics and Computation}, pages 158--179, 2021.

\bibitem[Brandt et~al.(2024)Brandt, Greger, Segal-Halevi, and Suksompong]{brandt2024optimal}
F.~Brandt, M.~Greger, E.~Segal-Halevi, and W.~Suksompong.
\newblock Optimal budget aggregation with single-peaked preferences.
\newblock \emph{arXiv:2402.15904}, 2024.

\bibitem[Cabannes(2004)]{cabannes2004participatory}
Y.~Cabannes.
\newblock Participatory budgeting: a significant contribution to participatory democracy.
\newblock \emph{Environment and urbanization}, 16\penalty0 (1):\penalty0 27--46, 2004.

\bibitem[Caragiannis et~al.(2016)Caragiannis, Procaccia, and Shah]{caragiannis2016truthful}
I.~Caragiannis, A.~Procaccia, and N.~Shah.
\newblock Truthful univariate estimators.
\newblock In \emph{Proceedings of the 33rd International Conference on Machine Learning}, pages 127--135, 2016.

\bibitem[Caragiannis et~al.(2022)Caragiannis, Christodoulou, and Protopapas]{caragiannis2022truthful}
I.~Caragiannis, G.~Christodoulou, and N.~Protopapas.
\newblock Truthful aggregation of budget proposals with proportionality guarantees.
\newblock In \emph{Proceedings of the 36th AAAI Conference on Artificial Intelligence}, pages 4917--4924, 2022.

\bibitem[Elkind et~al.(2023)Elkind, Suksompong, and Teh]{elkind2023settling}
E.~Elkind, W.~Suksompong, and N.~Teh.
\newblock Settling the score: Portioning with cardinal preferences.
\newblock In \emph{Proceedings of the 26th European Conference on Artificial Intelligence}, pages 621--628, 2023.

\bibitem[Fain et~al.(2016)Fain, Goel, and Munagala]{fain2016core}
B.~Fain, A.~Goel, and K.~Munagala.
\newblock The core of the participatory budgeting problem.
\newblock In \emph{Proceedings of the 12th International Conference on Web and Internet Economics}, pages 384--399, 2016.

\bibitem[Freeman and Schmidt-Kraepelin(2024)]{freeman2024project}
R.~Freeman and U.~Schmidt-Kraepelin.
\newblock Project-fair and truthful mechanisms for budget aggregation.
\newblock In \emph{Proceedings of the 38th {AAAI} Conference on Artificial Intelligence}, pages 9704--9712, 2024.

\bibitem[Freeman et~al.(2021)Freeman, Pennock, Peters, and Vaughan]{freeman2021truthful}
R.~Freeman, D.~M. Pennock, D.~Peters, and J.~W. Vaughan.
\newblock Truthful aggregation of budget proposals.
\newblock \emph{Journal of Economic Theory}, 193\penalty0 (3):\penalty0 105234, 2021.

\bibitem[Goel et~al.(2019)Goel, Krishnaswamy, Sakshuwong, and Aitamurto]{goel2019knapsack}
A.~Goel, A.~K. Krishnaswamy, S.~Sakshuwong, and T.~Aitamurto.
\newblock Knapsack voting for participatory budgeting.
\newblock \emph{ACM Transactions on Economics and Computation}, 7\penalty0 (2):\penalty0 1--27, 2019.

\bibitem[Goyal et~al.(2023)Goyal, Sakshuwong, Sarmasarkar, and Goel]{goyal2023low}
M.~Goyal, S.~Sakshuwong, S.~Sarmasarkar, and A.~Goel.
\newblock {Low Sample Complexity Participatory Budgeting}.
\newblock In \emph{Proceedings of the 50th International Colloquium on Automata, Languages, and Programming}, pages 70:1--70:20, 2023.

\bibitem[Jennings et~al.(2023)Jennings, Laraki, Puppe, and Varloot]{jennings2023new}
A.~B. Jennings, R.~Laraki, C.~Puppe, and E.~M. Varloot.
\newblock New characterizations of strategy-proofness under single-peakedness.
\newblock \emph{Mathematical Programming}, pages 1--32, 2023.

\bibitem[Lindner et~al.(2008)Lindner, Nehring, and Puppe]{lindner2008midpoint}
T.~Lindner, K.~Nehring, and C.~Puppe.
\newblock Allocating public goods via the midpoint rule.
\newblock In \emph{Proceedings of the 9th International Meeting of the Society for Social Choice and Welfare}, 2008.

\bibitem[Mass{\'o} and De~Barreda(2011)]{masso2011strategy}
J.~Mass{\'o} and I.~M. De~Barreda.
\newblock On strategy-proofness and symmetric single-peakedness.
\newblock \emph{Games and Economic Behavior}, 72\penalty0 (2):\penalty0 467--484, 2011.

\bibitem[Michorzewski et~al.(2020)Michorzewski, Peters, and Skowron]{michorzewski2020price}
M.~Michorzewski, D.~Peters, and P.~Skowron.
\newblock Price of fairness in budget division and probabilistic social choice.
\newblock In \emph{Proceedings of the 34th AAAI Conference on Artificial Intelligence}, pages 2184--2191, 2020.

\bibitem[Moulin(1980)]{moulin1980strategy}
H.~Moulin.
\newblock On strategy-proofness and single peakedness.
\newblock \emph{Public Choice}, pages 437--455, 1980.

\bibitem[Renault and Trannoy(2005)]{renault2005protecting}
R.~Renault and A.~Trannoy.
\newblock Protecting minorities through the average voting rule.
\newblock \emph{Journal of Public Economic Theory}, 7\penalty0 (2):\penalty0 169--199, 2005.

\bibitem[Renault and Trannoy(2011)]{renault2011assessing}
R.~Renault and A.~Trannoy.
\newblock Assessing the extent of strategic manipulation: the average vote example.
\newblock \emph{SERIEs}, 2:\penalty0 497--513, 2011.

\bibitem[Tao(2022)]{tao2022analysis}
T.~Tao.
\newblock \emph{Analysis {II}}.
\newblock Springer, 4 edition, 2022.

\bibitem[Wagner and Meir(2023)]{wagner2023strategy}
J.~Wagner and R.~Meir.
\newblock Strategy-proof budgeting via a vcg-like mechanism.
\newblock In \emph{Proceedings of the 16th International Symposium on Algorithmic Game Theory}, pages 401--418, 2023.

\end{thebibliography}
\bibliographystyle{abbrvnat}

\newpage
\appendix

\section{Missing Proofs of \Cref{sec:existence_truthful_non_phantom}}

\cutoffCAN*
\begin{proof}
    Let $\F$ be a family of phantom systems, fix $n, m \in \mathbb{N}$ and let $\tau \in (\sfrac{1}{2},1]$. 
    
    \textbf{$\A^\F_c$ is anonymous:}
    Given a profile $(\p_1, \dots, \p_n)$ and a permutation $\sigma: [n] \rightarrow [n]$, we know by anonymity of the moving-phantom mechanism $\A^{\F}$ (see \cref{thm:phantoms_anonymous_neutral_continuous}) that $$\A_c^\F(\p_1, \dots, \p_n) = c(\A^\F(\p_1, \dots, \p_n)) = c(\A^\F(\p_{\sigma(1)}, \dots, \p_{\sigma(n)})) = \A_c^\F(\p_{\sigma(1)}, \dots, \p_{\sigma(n)}).$$
    
    \textbf{$\A^\F_c$ is neutral:} 
    Given a profile $P = (\p_1, \dots, \p_n)$, with $\p_i = (\p_{i,1}, \dots, \p_{i,m})$, and a permutation $\sigma: [m] \rightarrow [m]$, let $P^\sigma = (\p_1^\sigma, \dots, \p_n^\sigma)$, with $\p_i^\sigma = (\p_{i,\sigma(1)}, \dots, \p_{i,\sigma(m)})$, as in the definition of neutrality. Let $a = \A^\G(P)$ and $a^\sigma = \A^\F(P^\sigma)$.
    Then, $$\A_c^\F(P^\sigma) = c(a^{\sigma}) = c(a_{\sigma(1)}, \dots, a_{\sigma(m)})= (c(a)_{\sigma(1)}, \dots, c(a)_{\sigma(m)}),$$ where the second equation follows from neutrality of \textsc{GreedyMax} (see \cref{thm:phantoms_anonymous_neutral_continuous}) and the last equation follows from neutrality of the cutoff function. Thus, \textsc{CutOffGreedyMax} is neutral.

    \textbf{$\A^\G_c$ is continuous:}
    One can easily verify that $c$ is a continuous function and $\A^\F$ is continuous by \cref{thm:phantoms_anonymous_neutral_continuous}. Then $\A^\F_c$ is continuous, since it is a composition of two continuous functions \citep[Corollary 2.1.7]{tao2022analysis}. 
\end{proof}

\subsection{Missing Proofs of \Cref{subsec:greedymax}} \label{sec:appthreeone}

In this section, we show that \textsc{CutoffGreedyMax} is truthful.
The proof distinguishes several cases based on whether \textsc{GreedyMax} returns an extreme aggregate for both the truthful report and a misreport of a voter. To support this case distinction, \Cref{lem:greedy_max_exact_aggregate} specifies the aggregate of \textsc{GreedyMax} in situations where all voter reports are (weakly) $(j,\tau)$-extreme for the same alternative $j$ and an additional condition on the other alternatives. \Cref{lem:greedy_max_high_aggregate} characterizes the profiles for which \textsc{GreedyMax} returns extreme outcomes.

\begin{lemma}\label{lem:greedy_max_exact_aggregate}
For all $n,m \in \N, m \ge 2$, $\tau \in [\frac{1}{2}, 1)$ and profiles $P\in \mathcal{P}_{n,m}$ for which $\min_{i \in [n]} p_{i,j} \geq \tau$ for some $j \in [m]$ and $\sum_{\substack{k \in [m] \\ k\neq j}} \max_{i \in [n]}(\p_{i,k}) \le 1-\tau$, the aggregate $a$ of \textsc{GreedyMax} is given as
    $$
        a_k =
        \begin{cases} 
            \max_{i \in [n]}(\p_{i,k}) & \text{ for } k \neq j \\
            1 - \sum_{\substack{h \in [m]\\ h \neq j}} \max_{i \in [n]}(\p_{i,h}) & \text{ for } k = j.
        \end{cases}
    $$
\end{lemma}
\begin{proof}
    At time $t = \tau$, \cref{obs:medians} implies that the medians on all alternatives besides $j$ equal $\max_{i \in [n]}(p_{i,k})$. Moreover, at time of normalization the median on the first alternative can be at most $\max_{i \in [n]}p_{i,j}$ since the sum of all maxima is at least $1$. Thus, we get that $t = 1 - \sum_{\substack{k \in [m]\\ k \neq j}} \max_{i \in [n]}(\p_{i,k}) \ge \tau$ is a time of normalization which implies the lemma statement. 
\end{proof}

The next lemma shows that \textsc{GreedyMax} returns an extreme aggregate if and only if \emph{all} voter reports are $(j,\tau)$-extreme for the same alternative~$j$ and the maximum votes on the other alternatives sum up to less than $1-\tau$.

\lemmatwo*
\begin{proof}
    We first show the forward direction of the equivalence.

    Assume for contradiction that $a_j > \tau$, but $\min_{i \in [n]}(\p_{i,j}) \le \tau$ or $\sum_{\substack{k \in [m] \\ k\neq j}} \max_{i \in [n]}(\p_{i,k}) \ge 1 - \tau$.
    \begin{CaseTree}
        \item There exists a vote that is not $j$-extreme, i.e., $\min_{i \in [n]}(\p_{i,j}) \le \tau$.\\
        Let $i^* = \argmin_{i \in [n]} (p_{i,j})$.
        Since $a_j > \tau$ we know by \Cref{obs:medians} that for any time $t$ of normalization $t>\tau$. Again by \Cref{obs:medians}, we know that $a_k$ is lower bounded by $\min\{t,p_{i^*,k}\}$ for all alternatives $k \neq j$. Thus, 
        $$
            \sum_{k \in [m]} a_k \ge a_j + \sum_{\substack{k \in [m]\\ k \neq j}} \min\{t,p_{i^*,k}\} > \tau + 1 - \tau =1,
        $$
        a contradiction.
        
        \item All votes are $j$-extreme but $\sum_{\substack{k \in [m] \\ k\neq j}} \max_{i \in [n]}(\p_{i,k}) \ge 1 - \tau$.\\
        Since $a_j > \tau$ we know by \Cref{obs:medians} that for any time $t$ of normalization $t > \tau \ge 1-\tau$. Again, by \cref{obs:medians} we know for each coordinate $k \neq j$ that $a_k = \min\{t, \max_{i \in [n]}(\p_{i,k})\} \ge \min\{1-\tau, \max_{i \in [n]}(\p_{i,k})\}$. But then 
        $$
            \sum_{k \in [m]} a_k \ge a_j + \sum_{\substack{k \in [m] \\ k\neq j}} \min\{1-\tau, \max_{i \in [n]}(\p_{i,k})\} > \tau + 1 - \tau = 1,
        $$
        again contradicting our initial assumption.
    \end{CaseTree}

    We now show the backwards direction of the equivalence.\\

    Assume $\min_{i \in [n]}(\p_{i,j}) > \tau$ and $\sum_{\substack{k \in [m] \\ k\neq j}} \max_{i \in [n]}(\p_{i,k}) < 1 - \tau$. Then we know from \Cref{lem:greedy_max_exact_aggregate} that $a_k = \max_{i \in [n]}(\p_{i,k})$ for all alternatives $k \neq j$. Then 
    \[
        a_j = 1 - \sum_{\substack{k \in [m] \\ k\neq j}} a_k = 1 - \sum_{\substack{k \in [m] \\ k\neq j}} \max_{i \in [n]}(\p_{i,k}) > \tau.\qedhere
    \]
    
\end{proof}

Making use of \Cref{lem:greedy_max_exact_aggregate,lem:greedy_max_high_aggregate}, we can now show the truthfulness of \textsc{CutoffGreedyMax}.

\maxcutoffgreedy*
\begin{proof}
    Let $n,m \in \N$ with $m \ge 2$, $\tau \in [\frac{1}{2},1)$ and $P = (\p_1, \dots, \p_n) \in \mathcal{P}_{n,m}$ be a profile. By the neutrality of $A^\F_c$ we can assume without loss of generality that $\p_{1,1} \le \dots \le \p_{1,m}$.

    Suppose there is a truthfulness violation for a voter. By the anonymity of $A^\F_c$, we can assume without loss of generality that this is voter $n$, i.e., there is $\p_n^\star \in \Delta^{(m-1)}$ such that for the profile $P^\star = (\p_1, \dots, \p_{n-1}, \p_n^\star)$ we have $\ellone{\p_n}{ \A^\F_c(P^\star)} < \ellone{\p_n}{\A^\F_c(P)}$.
    Let $a = \A^\F(P)$, $a^\star = \A^\F(P^\star)$, $\bar{a} = \A^\F_c(P)$ and $\bar{a}^\star = \A^\F_c(P^\star)$. In words, $a$ and $a^*$ are the aggregates of $\A^\F$ before and after the manipulation by voter $n$, respectively. Moreover, $\bar{a}$ and $\bar{a}^\star$ are the aggregates of $A^\F_c$ before and after the manipulation by voter $n$, respectively. 
    
    For all $j \in [m-1]$, we know that $\p_{1,j} \leq \tau$ and thus, by \Cref{lem:greedy_max_high_aggregate}, $a_j \le \tau$ and $a_j^\star \le \tau$. In other words, if $a$ or $a^*$ is extreme, then they are $m$-extreme. 

    \begin{CaseTree}
        \item Neither $a$ nor $a^*$ is $m$-extreme. 
        Then, $a = \bar{a}$, $a^\star = \bar{a}^\star$ and $$\ellone{\p_n}{\A^\F(P^*)} = \ellone{\p_n}{a^*} = \ellone{\p_n}{ \bar{a}^\star} < \ellone{\p_n}{\bar{a}} = \ellone{\p_n}{a} = \ellone{\p_n}{\A^\F(P)},$$ 
        which implies that $\A^\F$ is not truthful, contradicting the truthfulness of moving-phantom mechanisms (\Cref{thm:phantoms_anonymous_neutral_continuous}).
        \item The aggregate $a$ is not $m$-extreme but $a^*$ is $m$-extreme. \\
        Then, $\bar{a} = a$. 
        We argue that the aggregate $\bar{a}^\star$ can be achieved through a misreport by voter $n$ when using $\A^\F$, which leads to a contradiction to the truthfulness of $\A^\F$.

        Since $a^*$ is $m$-extreme, we know by \Cref{lem:greedy_max_high_aggregate} that $p_{i,m}>\tau$ for all $i \in [n-1]$, $p^*_{n,m}>\tau$, and $\sum_{j \in [m-1]}\max\{p^*_{n,j},\max_{i \in [n-1]}p_{i,j}\} < 1 - \tau$. Thus, by \Cref{lem:greedy_max_exact_aggregate} this implies that $a^*_j = \max\{p^*_{n,j},\max_{i \in [n-1]}p_{i,j}\}$ for all $j \in [m-1]$. Moreover, by the definition of the cutoff function, it holds that $\bar{a}^*_j > a^*_j$ for all $j \in [m-1]$ and $\bar{a}^*_m=\tau$. 

        We build the profile $Q = (\p_1, \dots, \p_{n-1}, \bar{a}^\star)$ and define $b = \A^\F(Q)$. We want to apply \Cref{lem:greedy_max_exact_aggregate} again. We can do so since $\bar{a}_m^* = \tau$, all voters in $[n-1]$ are $m$-extreme, and the last condition of \Cref{lem:greedy_max_exact_aggregate} is satisfied as well, namely, $\sum_{j \in [m-1]}\max\{\bar{a}^*_j, \max_{i \in [n-1]} p_{i,j}\}= \sum_{j \in [m-1]}\bar{a}^*_j = 1 - \tau$. Thus, we get $$b_j = \max \{\bar{a}^\star_j, \max_{i \in [n-1]} p_{i,j}\} = \bar{a}^\star_j \text{ for all } j \in [m-1].$$ Thus, also $b_m = \bar{a}^\star_m$ and $b = \bar{a}^\star$. 
        This contradicts the truthfulness of $\A^\F$, since $$\ellone{\p_n}{\A^\F(Q)} = \ellone{\p_n}{b} = \ellone{\p_n}{\bar{a}^*} < \ellone{\p_n}{\bar{a}} = \ellone{\p_n}{a} = \ellone{\p_n}{\A^\F(P)}.$$

        \item The aggregate $a$ is $m$-extreme.\\
        By \Cref{lem:greedy_max_high_aggregate}, all reports in $P$ are $m$-extreme and $\sum_{j \in [m-1]} \max_{i \in [n]}\{p_{i,j}\} < 1-\tau$. Thus, by \Cref{lem:greedy_max_exact_aggregate}, $a_j = \max_{i \in [n]} (p_{i,j})\ge \p_{n,j}$ for all $j \in [m-1]$. Therefore, $\bar{a}_m = \tau < \p_{n,m}$ and $\bar{a}_j > a_j \geq \p_{n,j}$ for all $j \in [m-1]$. Now, by assumption $\ellone{\p_n}{\bar{a}^\star} < \ellone{\p_n}{\bar{a}}$ and thus
        \begin{align*}
            \ellone{\p_n}{\bar{a}^\star} - \ellone{\p_n}{\bar{a}} &=  
            2 - 2\sum_{j=1}^m \min(\p_{n,j}, \bar{a}^\star_j) - \big(2 - 2\sum_{j=1}^m \min(\p_{n,j}, \bar{a}_j)\big)\\ 
            &= 2 \sum_{j=1}^m (\min(\p_{n,j}, \bar{a}_j) - \min(\p_{n,j}, \bar{a}^\star_j)) \\ 
            &= 2 \big(\bar{a}_m - \min(\p_{n,m}, \bar{a}^\star_m) + 2 \sum_{j=1}^{m-1} \underbrace{(\p_{n,j} - \min(\p_{n,j}, \bar{a}^\star_j))}_{\ge 0}\big) &< 0,
        \end{align*}
        where the first equality holds because of the simple observation\footnote{ 
        \label{footnoteLemma}
            $
               \ellone{u}{v} 
               = \sum\limits_{i=1}^m (\max(u_i, v_i) - \min(u_i, v_i)) 
               = \sum\limits_{i=1}^m (\max(u_i, v_i) + \min(u_i, v_i)) - 2 \sum\limits_{i=1}^m \min(u_i, v_i) 
               = 2 - 2 \sum\limits_{i=1}^m \min(u_i, v_i) 
            $
        } that for two vectors $u, v \in \Delta^{(m-1)}$ in the simplex $\ellone{u}{v} = 2 - 2\sum_{i=1}^m \min(u_i, v_i)$. But then $\bar{a}^\star_m > \bar{a}_m = \tau$, contradicting the definition of $\A^\F_c$. \qedhere
    \end{CaseTree}\end{proof}

\subsection{Missing Proofs of \Cref{sec:moretruthful}} \label{sec:appthreetwo}
In order to prove the propositions and lemmas from \Cref{sec:moretruthful}, we first formally introduce the phantom mechanisms \textsc{MaxUtilitarianWelfare}, \textsc{IndependentMarkets}, \textsc{PiecewiseUniform} and \textsc{Ladder}.

The \textsc{MaxUtilitarianWelfare} mechanism \citep{freeman2021truthful} is the moving-phantom mechanism induced by the family of phantom systems $\F = \{\F_n \mid n \in \N\}$, with $\F_n = \{f_0, \dots, f_{n}\}$, where $f_k$ is defined as 
$$
    f_k(t) = 
    \begin{cases}
        0 & \text{ if } t < \frac{k}{n+1}, \\
        t(n+1)-k & \text{ if } \frac{k}{n+1} \le t \le \frac{k+1}{n+1}, \\
        1 & \text{ if } \frac{k+1}{n+1} < t.
    \end{cases}
$$
This corresponds to moving each phantom from zero to one consecutively.

In the following definitions, we again allow a slight deviation from the definition of a phantom system, by defining phantom functions $f$ with $f(1) < 1$, as long as normalization is guaranteed for all profiles, for $t \le 1$.
The \textsc{IndependentMarkets} mechanism \citep{freeman2021truthful} is the moving-phantom mechanism induced by the family of phantom systems defined as $$\F_n = \{f_k(t) = \max(t \cdot \frac{n-k}{n}, 0) \text{ for } k = 0, \dots, n\}$$ for all $n \in \N$.
This corresponds to moving all phantoms simultaneously but at different speeds, until each phantom $k$ reaches its final position $\frac{n-k}{n}$. 

The \textsc{PiecewiseUniform} mechanism \citep{caragiannis2022truthful} is the moving-phantom mechanism induced by the family of phantom systems $\F = \{\F_n \mid n \in \N\}$, with $\F_n = \{f_0, \dots, f_{n}\}$, where $f_k$ is defined as 
$$
    f_k(t) = 
    \begin{cases}
        \frac{4t(n-k))}{n}-2t  & \text{ if } \frac{k}{n} \le \frac{1}{2}, \\
        0 & \text{ if } \frac{k}{n} > \frac{1}{2}, 
    \end{cases}
$$
if $t < \frac{1}{2}$ and  
$$
    f_k(t) = 
    \begin{cases}
        \frac{(n-k) (3-2t)}{n} -2 + 2t  & \text{ if } \frac{k}{n} \le \frac{1}{2}, \\
        \frac{(n-k) (2t-1)}{n} & \text{ if } \frac{k}{n} > \frac{1}{2},
    \end{cases}
$$ if $t \ge \frac{1}{2}$.

The \textsc{Ladder} mechanism \citep{freeman2024project} is the moving-phantom mechanism induced by the family of phantom systems defined as $$ \F_n = \{f_k(t) = \max(t - \frac{k}{n}, 0) \text{ for } k = 0, \dots, n\}$$ for all $n \in \N$.

\noConstantThreshold*
\begin{proof}
    We first show the truthfulness violation for \textsc{CutoffMaxUtilitarianWelfare} and then give a unified proof for \textsc{CutoffIM}, \textsc{CutoffPiecewiseUniform} and \textsc{CutoffLadder}.

    \textbf{\textsc{CutoffMaxUtilitarianWelfare}:}
    
    Let $\tau \in [\frac{1}{2}, 1)$.
    We choose $m = 3$ and $n = 2$ and consider the profile $$P = \Big(\Big(\frac{1+3\tau}{4}, \frac{3-3\tau}{4}, 0\Big),\Big(\frac{1+\tau}{2}, \frac{1-\tau}{4}, \frac{1-\tau}{4}\Big)\Big).$$
    \textsc{MaxUtilitarianWelfare} reaches normalization, when the second phantom reaches position $\frac{1-\tau}{2}$ (at this time the first phantom is at position 1 and the third one at position 0). Thus, the medians form the vector $(\frac{1+3\tau}{4}, \frac{1-\tau}{2}, \frac{1-\tau}{4})$, which is normalized and thus the outcome of \textsc{MaxUtilitarianWelfare}. \textsc{CutoffMaxUtilitarianWelfare} therefore returns the aggregate $a = (\tau, \frac{5-5\tau}{8}, \frac{3-3\tau}{8})$. Now, consider the profile $$P^\star = \Big(\Big(\tau, 1-\tau, 0\Big),\Big(\frac{1+\tau}{2}, \frac{1-\tau}{4}, \frac{1-\tau}{4}\Big)\Big),$$ in which the first voter misreported their preference. Here, \textsc{MaxUtilitarianWelfare} reaches normalization, when the second phantom reaches $\frac{3-3\tau}{4}$ and the medians are $(\tau, \frac{3-3\tau}{4}, \frac{1-\tau}{4})$. Since the first coordinate is equal to $\tau$, the cutoff has no effect and \textsc{CutoffMaxUtilitarianWelfare} returns the aggregate $a^\star = (\tau, \frac{3-3\tau}{4}, \frac{1-\tau}{4})$. This is a truthfulness violation, since
    $$\textstyle
        \ellone{(\frac{1+3\tau}{4}, \frac{3-3\tau}{4}, 0)}{a^\star} = \frac{1-\tau}{2} < \frac{5-5\tau}{8} = \ellone{(\frac{1+3\tau}{4}, \frac{3-3\tau}{4}, 0)}{a}.
    $$

    \textbf{\textsc{CutoffIM}, \textsc{CutoffPiecewiseUniform} and \textsc{CutoffLadder}:}

    Let $\tau \in [\frac{1}{2}, 1)$.
    We choose $m = 3$ and $n$, such that $\frac{n-2}{n} \le \tau < \frac{n-1}{n}$. Note, that since $\tau \ge \frac{1}{2}$, we have $n \ge 3$. Now, consider the profile $P = \{p_1, \dots, p_n\} \in \mathcal{P}_{n,m}$ with $p_1 = (0, 1-\eps, \eps)$ and $p_2 = \dots = p_n = (1, 0, 0)$ for some $\eps > 0$. For \textsc{IndependentMarkets} we choose $\eps = \frac{(1-\tau)n-1}{3n-1}$, for \textsc{Ladder} $\eps = 2\frac{(1-\tau)n-1}{5n}$ and for \textsc{PiecewiseUniform} $\eps = 2\frac{(1-\tau)n-1}{5n}$.

    \textbf{Claim:}
    Then the aggregate $\A^\F(P)$ for any of these phantoms is $\A^\F(P) = (\tau + 2\eps, 1-\tau-3\eps, \eps)$ and thus $\bar{a} = \A^\F_c(P) = (\tau, 1-\tau-2\eps, 2\eps)$. 

    \begin{proof_of_claim}
        We will proof the claim for each mechanism independently. \\
        \textbf{\textsc{IndependentMarkets:}} Note, that
            \begin{equation} \label{eq:im_eps1}
                \eps = \frac{(1-\tau)n-1}{3n-1} \le \frac{(1-\frac{n-2}{n})n-1}{3n-1} = \frac{1}{3n-1} < \frac{1}{n+1}
            \end{equation}
            and also by the definition of $\eps$
            \begin{equation} \label{eq:im_eps2}
                \tau = 1 - \frac{\eps(3n-1)+1}{n} = \frac{n-1}{n} - \eps (2+\frac{n-1}{n}) = \frac{n-1}{n} \cdot (1 - \eps) - 2 \eps.
            \end{equation}
            \textsc{IndependentMarkets} reaches normalization at time $t = 1-\eps$, where we have $f_1(1-\eps) = \frac{n-1}{n}(1-\eps)$, $f_{n-1}(1-\eps) = \frac{1}{n}(1-\eps)$ and $f_n(1-\eps) = 0$. On the first coordinate, the median is $f_1(1-\eps)$, on the second it is $f_1(1-\eps)$ and on the third it is $\eps$, since $f_1(1-\eps) = \frac{1}{n}(1-\eps) > \eps$ by \Cref{eq:im_eps1} and $f_n(1-\eps) = 0$.
            Thus, the aggregate returned is $a = (\frac{n-1}{n}(1-\eps), \frac{1}{n}(1-\eps), \eps)$, which is normalized and thus equals $a = (\tau + 2\eps, 1-\tau-3\eps, \eps)$ by \Cref{eq:im_eps2}.
    
        \textbf{\textsc{PiecewiseUniform:}} Note, that
            \begin{equation} \label{eq:pu_eps1}
                \eps = 2\frac{(1-\tau)n-1}{5n} \le 2\frac{(1-\frac{n-2}{n})n-1}{5n} = \frac{2}{5n} < \frac{2}{3n}
            \end{equation} 
            and also by the definition of $\eps$
            \begin{equation} \label{eq:pu_eps2}
                \tau = \frac{n-1}{n} - \frac{5\eps}{2} = \frac{n-1}{n} - \frac{\eps}{2} - 2\eps.
            \end{equation}
            \textsc{PiecewiseUniform} reaches normalization at time $t^\star = 1 - \frac{n\eps}{4} > \frac{1}{2}$, where we have $f_1(t^\star) = \frac{n-1}{n} - \frac{\eps}{2}$, $f_{n-1}(t^\star) = \frac{1}{n} - \frac{\eps}{2}$ and $f_n(t^\star) = 0$. On the first coordinate, the median is $\frac{n-1}{n} - \frac{\eps}{2}$, on the second it is $\frac{1}{n} - \frac{\eps}{2}$ and on the third it is $\eps$, since $f_1(t^\star) = \frac{1}{n} - \frac{\eps}{2} > \eps$ by \Cref{eq:pu_eps1} and $f_n(t^\star) = 0$.
            Thus, the aggregate returned is $a = (\frac{n-1}{n} - \frac{\eps}{2}, \frac{1}{n} - \frac{\eps}{2}, \eps)$, which is normalized and thus equals $a = (\tau + 2\eps, 1-\tau-3\eps, \eps)$ by \Cref{eq:pu_eps2}.
            
        \textbf{\textsc{Ladder:}} Note, that \Cref{eq:pu_eps1,eq:pu_eps2} hold, as shown above.
            \textsc{Ladder} reaches normalization at time $t^\star = 1 - \frac{\eps}{2}$, where we have $f_{n-1}(t^\star) = \frac{n-1}{n} - \frac{\eps}{2}$, $f_{n-1}(t^\star) = \frac{1}{n} - \frac{\eps}{2}$ and $f_n(t^\star) = 0$. On the first coordinate, the median is $\frac{n-1}{n} - \frac{\eps}{2}$, on the second it is $\frac{1}{n} - \frac{\eps}{2}$ and on the third it is $\eps$, since $f_1(t^\star) = \frac{1}{n} - \frac{\eps}{2} > \eps$ by \Cref{eq:pu_eps1} and $f_n(t^\star) = 0$.
            Thus, the aggregate returned is $a = (\frac{n-1}{n} - \frac{\eps}{2}, \frac{1}{n} - \frac{\eps}{2}, \eps)$, which is normalized and thus equals $a = (\tau + 2\eps, 1-\tau-3\eps, \eps)$ by \Cref{eq:pu_eps2}. \qedhere
    \end{proof_of_claim}

    Now, consider the profile $P^\star = (p_1^\star, p_2, \dots, p_n)$, where voter 1 misreports $p_1^\star = (0,1,0)$. All three phantom mechanisms return the aggregate $a^\star = \A^\F(P^\star) = (\frac{n-1}{n}, \frac{1}{n}, 0)$ and thus the cutoff-mechanisms return $\bar{a}^\star = \A^\F_c(P^\star) = (\tau, \frac{1}{n} + \frac{1}{2}(\frac{n-1}{n}-\tau), \frac{1}{2}(\frac{n-1}{n}-\tau))$.
    
    \textbf{Claim:}
    Then $\frac{1}{2}(\frac{n-1}{n}-\tau) < 2\eps$. 
    
    \begin{proof_of_claim}
        We will proof the claim for the mechanisms independently. \\
        \textbf{\textsc{IndependentMarkets:}} 
            By \Cref{eq:im_eps2} we have
            \begin{equation*}
                \frac{1}{2}(\frac{n-1}{n}-\tau) = \frac{1}{2}(\frac{n-1}{n}-\frac{n-1}{n} \cdot (1 - \eps) + 2 \eps) = \frac{1}{2}(\eps \frac{n-1}{n} + 2 \eps) < 2\eps
            \end{equation*}
        \textbf{\textsc{PiecewiseUniform} and \textsc{Ladder:}}
            By \Cref{eq:pu_eps2} we have
            \begin{equation*}
                \frac{1}{2}(\frac{n-1}{n}-\tau) = \frac{1}{2}(\frac{n-1}{n} - \frac{n-1}{n} + \frac{\eps}{2} + 2\eps) = \frac{\eps}{4} + \eps < 2\eps
            \end{equation*}
            \qedhere
    \end{proof_of_claim}
    
    Since $\bar{a}^\star$ is normalized, we have $\bar{a}^\star = (\tau, 1 - \tau - \frac{1}{2}(\frac{n-1}{n}-\tau)), \frac{1}{2}(\frac{n-1}{n}-\tau))$. This is a truthfulness violation, since    
    \begin{align*}
        \ellone{p_1}{\bar{a}^\star} 
        &= \tau + (1-\eps)-(1 - \tau - \underbrace{\frac{1}{2}(\frac{n-1}{n}-\tau)}_{<2\eps}) + |\underbrace{\frac{1}{2}(\frac{n-1}{n}-\tau)}_{<2\eps} - \eps| \\
        &< \tau + (1-\eps)-(1 - \tau - 2\eps) + \eps = \ellone{p_1}{\bar{a}}. \qedhere
    \end{align*}

\end{proof}

We now want to work towards showing \Cref{prop:phantom_cutoff_truthfulness}, which states that for any slow phantom mechanism $\A^\F$, we can find a threshold function $\tau: \N \times \N \rightarrow [0, 1)$, such that its cutoff version $\A^\F_{c_\tau}$ is truthful and not a phantom mechanism itself. The proof very closely follows the proofs of \Cref{prop:greedy_max_cutoff_truthfulness,prop:greedy_max_cutoff_not_phantom} that state that \textsc{CutoffGreedyMax} is truthful and not a phantom for any $\tau \in [\frac{1}{2}, 1)$, respectively. We first explain how to find the value of $\tau$ for any given slow phantom mechanism and values $m,n \in \N$ and then state \Cref{lem:phantom_exact_aggregate}, which can be seen as a slightly strengthened version of \Cref{lem:greedy_max_exact_aggregate} applied to slow phantom mechanisms and the right value of $\tau$. Similarly, we can then show \Cref{lem:phantom_high_aggregate} which has the same statement as \Cref{lem:greedy_max_high_aggregate}, the only difference being that it talks about slow phantoms and thus needs a stronger condition for the choice of $\tau$. Finally, the proof of truthfulness will be exactly the same as the proof of \Cref{prop:greedy_max_cutoff_truthfulness}, referring to \Cref{lem:phantom_exact_aggregate} and \Cref{lem:phantom_high_aggregate} instead of \Cref{lem:greedy_max_exact_aggregate} and \Cref{lem:greedy_max_high_aggregate}.

In order to prove \Cref{lem:phantom_exact_aggregate}, given $m,n \in \N, m \ge 2$ and a phantom system $\F_n = \{f_0, \dots, f_n\}$, we want to find the maximum time $t^\star_{m,n}$, such that $f_0(t^\star_{m,n}) \le 1 - (m-1) \cdot f_{n-1}(t^\star_{m,n})$. The existence of $t$ is given, as any phantom system trivially satisfies this condition for $t = 0$. For a family of phantom system $\F = \{\F_n \mid n \in \N\}$, we define the \textit{threshold function} $\tau : \N \times \N \rightarrow [\frac{1}{2},1]$ as $\tau(n,m) = 1-f_{n-1}(t^\star_{m,n})$. We are in particular interested in the families of phantom systems that have $\tau(n,m) < 1$ for all $m,n \in \N, m \ge 2$.
By the following argument, we can see that this is the case for all slow families of phantom systems. Let $m \in \N, m \ge 2$ and $\F = \{\F_n \mid n \in \N\}$ be a slow family of phantom systems. Since $\F$ is slow, for any $n$, we know there is a time $t$, where $f_0(t) < 1$ and $f_{n-1}(t) > 0$. Then, there must be a time $t' \le t$, where we have $0 < f_{n-1}(t') \le \frac{1-f_0(t)}{m-1}$, at which we meet the above condition, since $f_0(t') \le f_0(t) \le 1 - (m-1) f_{n-1}(t')$ and $f_0(t') = f_0(t) = 0$. Thus, we have $\tau(n,m) = 1-f_{n-1}(t^\star_{m,n}) \le 1-f_{n-1}(t') < 1$.

\begin{lemma}\label{lem:phantom_exact_aggregate}
    For any slow phantom mechanism $\A^\F$ and its threshold function $\tau: \N \times \N \rightarrow [0,1)$ the following holds. For all $n,m \in \N$, $m \ge 2$ and any profile $P = (\p_1, \dots, \p_n) \in \mathcal{P}_{n,m}$ for which $p_{i,j} \geq \tau(n,m)$ for some $j \in [m]$ and all $i \in [n]$, the aggregate $a = \A^\F(P)$ is given as
    $$
        a_k =
        \begin{dcases}
            \max_{i \in [n]}(\p_{i,k}) & \text{ for } k \neq j \\
            1 - \sum_{\substack{h \in [m]\\ h \neq j}} \max_{i \in [n]}(\p_{i,h}) & \text{ for } k = j.
        \end{dcases}
    $$
\end{lemma}
\begin{proof}
    Let $t^\star$ be the minimum time at which $f_{n-1}(t^\star) = 1-\tau(n,m)$ and $t'$ the minimum time at which $f_0(t') = 1$. By the slowness of the phantom mechanism, we know that $f_0(t^\star) \le 1-(m-1) \cdot (1-\tau(n,m))$ and $f_n(t') = 0$.
    For each alternative $k \neq j$ we have $\max_{i \in [n]}(p_{i,k}) \le 1-\tau(n,m)$. Consider the position of the phantoms at time $t$, where $f_0(t) = 1 - \sum_{\substack{h \in [m]\\ h \neq j}} \max_{i \in [n]}(\p_{i,h})$. Note, that $t \in [t^\star, t']$, since 
    $$
        1-(m-1)\cdot (1-\tau(n,m))  = 1-\sum_{\substack{h \in [m]\\ h \neq j}}  1-\tau(n,m) \le 1-\sum_{\substack{h \in [m]\\ h \neq j}} \max_{i \in [n]}(\p_{i,h}).
    $$ Then exactly $n$ phantoms are at positions larger than $1-\tau(n,m) \ge \max_{i \in [n]} p_{i,k}$ for $k \neq j$ and one phantom is at position $0$, thus the median on alternative $k \neq j$ equals $\max_{i \in [n]} (p_{i,k})$. On alternative $j$, we have $n$ voters at positions larger than $\tau(n,m)$, thus the top phantom at position $1 - \sum_{\substack{h \in [m]\\ h \neq j}} \max_{i \in [n]}(\p_{i,h}) \le \min_{i \in [n]}(\p_{i,j})$ is the median and normalization is attained.
\end{proof}

\begin{lemma} \label{lem:phantom_high_aggregate}
    For any slow phantom mechanism $\A^\F$ and its threshold function $\tau: \N \times \N \rightarrow [0,1)$ the following holds. For all $n,m \in \N$, $m \ge 2$ and any profile $P = (\p_1, \dots, \p_n) \in \mathcal{P}_{n,m}$ and corresponding aggregate $a = \A^\F(P)$ it holds that $a_j > \tau(n,m)$ if and only if
    $$\min_{i \in [n]}(\p_{i,j}) > \tau(n,m) \qquad \textnormal{and} \qquad \sum_{\substack{k \in [m] \\ k\neq j}} \max_{i \in [n]}(\p_{i,k}) < 1-\tau(n,m).$$
\end{lemma}

\begin{proof}
    We first show the forward direction of the equivalence.

    Assume for contradiction that $a_j > \tau(n,m)$, but $\min_{i \in [n]}(\p_{i,j}) \le \tau(n,m)$ or $\sum_{\substack{k \in [m] \\ k\neq j}} \max_{i \in [n]}(\p_{i,k}) \ge 1-\tau(n,m)$.
    \begin{CaseTree}
        \item There exists a vote that is not $j$-extreme, i.e., $\min_{i \in [n]}(\p_{i,j}) \le \tau(n,m)$.\\
        Let $i^*$ be a voter with a vote that is not $j$-extreme, i.e.\ $\p_{i^\star,j} < \tau(n,m)$.
        Since $a_j > \tau(n,m)$ we know by the definition of the threshold function $\tau$ that for any time $t$ of normalization $f_0(t) > \tau(n,m) \ge 1-(m-1) \cdot (1-\tau(n,m))$. Thus, we also know that $f_{n-1}(t) \ge 1-\tau(n,m)$ and therefore $a_k$ is lower bounded by $\min\{1-\tau(n,m),p_{i^*,k}\}$ for all alternatives $k \neq j$. Thus, $$a_j + \sum_{\substack{k \in [m]\\ k \neq j}} a_k \geq a_j + \sum_{\substack{k \in [m]\\ k \neq j}} \min\{1-\tau(n,m),p_{i^*,k}\} > \tau(n,m) + 1-\tau(n,m) =1,$$
        where the last inequality holds, because $\sum_{\substack{k \in [m]\\ k \neq j}} p_{i^*,k} > 1-\tau(n,m)$. This is a contradiction to normalization at time $t$.
        
        \item All votes are $j$-extreme but $\sum_{\substack{k \in [m] \\ k\neq j}} \max_{i \in [n]}(\p_{i,k}) \ge 1-\tau(n,m)$.\\
        Then for all voters $i$ we have $\p_{i,j} > \tau(n,m)$ and thus $\p_{i,k} < 1-\tau(n,m)$ for all other alternatives $k \neq j$. From \Cref{lem:phantom_exact_aggregate}, we know that $a_k = \max_{i \in [n]}(\p_{i,k})$ for alternatives $k \neq j$, and thus $$a_j = 1 - \sum_{\substack{k \in [m] \\ k\neq j}} a_k = 1 - \sum_{\substack{k \in [m] \\ k\neq j}} \max_{i \in [n]}(\p_{i,k}) \le \tau(n,m),$$
        again contradicting our initial assumption.
    \end{CaseTree}

    We now show the backwards direction of the equivalence.\\

    Assume $\min_{i \in [n]}(\p_{i,j}) > \tau(n,m)$ and $\sum_{\substack{k \in [m] \\ k\neq j}} \max_{i \in [n]}(\p_{i,k}) < 1-\tau(n,m)$. Again, we know from \Cref{lem:phantom_exact_aggregate} that $a_k = \max_{i \in [n]}(\p_{i,k})$ for all alternatives $k \neq j$. Then 
    \[
        a_j = 1 - \sum_{\substack{k \in [m] \\ k\neq j}} a_k = 1 - \sum_{\substack{k \in [m] \\ k\neq j}} \max_{i \in [n]}(\p_{i,k}) > \tau(n,m).\qedhere
    \]
    
\end{proof}

We are now ready to prove \Cref{prop:phantom_cutoff_truthfulness} by using the threshold function $\tau$ as the cutoff for any given slow phantom mechanism. Recall, that for given $n, m \in \N$ and a phantom system $\F_n$, the output $\tau(n,m)$ of the cutoff function is defined as $1-f_{n-1}(t^\star_{m,n})$, where $t^\star_{m,n}$ is the maximum time at which $f_0(t^\star_{m,n}) \le 1 - (m-1) \cdot f_{n-1}(t^\star_{m,n})$.
While the arguments of the proof follow the same arguments as the respective proofs for \textsc{CutoffGreedyMax} we state the full proof for completeness.

\propCutOffPhantomsTruthful*

\begin{proof}
    Let $\F$ be a slow family of phantom systems. We show the theorem using the phantoms threshold function $\tau$, as defined above. Let $n,m \in \N$ with $m \ge 3$. Since $m$ and $n$ are fixed, we write $\tau$ instead of $\tau(n,m)$ for the remainder of the proof to improve readability. 

    \textbf{$\A^\F_{c_\tau}$ is truthful:}
    Let $P = (\p_1, \dots, \p_n) \in \mathcal{P}_{n,m}$ be a profile. By the neutrality of $A^\F_c$ we can assume without loss of generality that $\p_{1,1} \le \dots \le \p_{1,m}$.

    Suppose there is a truthfulness violation for a voter. By the anonymity of $A^\F_c$, we can assume without loss of generality that this is voter $n$, i.e. there is $\p_n^\star \in \Delta^{(m-1)}$ such that for the profile $P^\star = (\p_1, \dots, \p_{n-1}, \p_n^\star)$ we have $\ellone{\p_n}{ \A^\F_c(P^\star)} < \ellone{\p_n}{\A^\F_c(P)}$.
    Let $a = \A^\F(P)$, $a^\star = \A^\F(P^\star)$, $\bar{a} = \A^\F_c(P)$ and $\bar{a}^\star = \A^\F_c(P^\star)$. In words, $a$ and $a^*$ are the aggregates of $\A^\F$ before and after the manipulation by voter $n$, respectively. Moreover, $\bar{a}$ and $\bar{a}^\star$ are the aggregates of $A^\F_c$ before and after the manipulation by voter $n$, respectively. 
    
    For all $j \in [m-1]$, we know that $\p_{1,j} \leq \tau$ and thus, by \Cref{lem:phantom_high_aggregate}, $a_j \le \tau$ and $a_j^\star \le \tau$. In other words, if $a$ or $a^*$ is extreme, then they are $m$-extreme. 

    \begin{CaseTree}
        \item Neither $a$ nor $a^*$ is $m$-extreme. 
        Then, $a = \bar{a}$, $a^\star = \bar{a}^\star$ and $$\ellone{\p_n}{\A^\F(P^*)} = \ellone{\p_n}{a^*} = \ellone{\p_n}{ \bar{a}^\star} < \ellone{\p_n}{\bar{a}} = \ellone{\p_n}{a} = \ellone{\p_n}{\A^\F(P)},$$ 
        which implies that $\A^\F$ is not truthful, contradicting the truthfulness of moving-phantom mechanisms (\Cref{thm:phantoms_anonymous_neutral_continuous}).
        \item The aggregate $a$ is not $m$-extreme but $a^*$ is $m$-extreme. \\
        Then, $\bar{a} = a$. 
        We argue that the aggregate $\bar{a}^\star$ can be achieved through a misreport by voter $n$ when using $\A^\F$, which leads to a contradiction to the truthfulness of $\A^\F$.

        Since $a^*$ is $m$-extreme, we know by \Cref{lem:phantom_high_aggregate} that $p_{i,m}>\tau$ for all $i \in [n-1]$ and $p^*_{n,m}>\tau$. Thus, by \Cref{lem:phantom_exact_aggregate} this implies that $a^*_j = \max\{p^*_{n,j},\max_{i \in [n-1]}p_{i,j}\}$ for all $j \in [m-1]$. Moreover, by the definition of the cutoff function, it holds that $\bar{a}^*_j > a^*_j$ for all $j \in [m-1]$ and $\bar{a}^*_m=\tau$. 

        We build the profile $Q = (\p_1, \dots, \p_{n-1}, \bar{a}^\star)$ and define $b = \A^\F(Q)$. Since $\bar{a}_m^* \geq \tau$ and all voters in $[n-1]$ are $m$-extreme, we can again apply \Cref{lem:phantom_exact_aggregate} and get that $$b_j = \max \{\bar{a}^\star_j, \max_{i \in [n-1]} p_{i,j}\} = \bar{a}^\star_j \text{ for all } j \in [m-1].$$ Thus, also $b_m = \bar{a}^\star_m$ and $b = \bar{a}^\star$. 
        This contradicts the truthfulness of $\A^\F$, since $$\ellone{\p_n}{\A^\F(Q)} = \ellone{\p_n}{b} = \ellone{\p_n}{\bar{a}^*} < \ellone{\p_n}{\bar{a}} = \ellone{\p_n}{a} = \ellone{\p_n}{\A^\F(P)}.$$

        \item The aggregate $a$ is $m$-extreme.\\
        By \Cref{lem:phantom_high_aggregate}, all reports in $P$ are $m$-extreme. Thus, by \Cref{lem:phantom_exact_aggregate}, $a_j = \max_{i \in [n]} (p_{i,j})\ge \p_{n,j}$ for all $j \in [m-1]$. Therefore, $\bar{a}_m = \tau < \p_{n,m}$ and $\bar{a}_j > a_j \geq \p_{n,j}$ for all $j \in [m-1]$. Now, by assumption $\ellone{\p_n}{\bar{a}^\star} < \ellone{\p_n}{\bar{a}}$ and thus
        \begin{align*}
            \ellone{\p_n}{\bar{a}^\star} - \ellone{\p_n}{\bar{a}} &=  
            2 - 2\sum_{j=1}^m \min(\p_{n,j}, \bar{a}^\star_j) - \big(2 - 2\sum_{j=1}^m \min(\p_{n,j}, \bar{a}_j)\big)\\ 
            &= 2 \sum_{j=1}^m (\min(\p_{n,j}, \bar{a}_j) - \min(\p_{n,j}, \bar{a}^\star_j)) \\ 
            &= 2 \big(\bar{a}_m - \min(\p_{n,m}, \bar{a}^\star_m) + 2 \sum_{j=1}^{m-1} \underbrace{(\p_{n,j} - \min(\p_{n,j}, \bar{a}^\star_j))}_{\ge 0}\big) &< 0,
        \end{align*}
        where the first equality holds due to \Cref{footnoteLemma}. But then $\bar{a}^\star_m > \bar{a}_m = \tau$, contradicting the definition of $\A^\F_c$. \qedhere
    \end{CaseTree}

    \textbf{$\A^\F_{c_\tau}$ is not a phantom mechanism:}
    Consider the profile $P \in \mathcal{P}_{n,m}$ defined as $$P = \Big(\Big(1, 0, \dots,0\Big), \dots, \Big(1, 0, \dots,0\Big),\Big(\frac{1+\tau}{2}, \frac{1-\tau}{2}, 0, \dots, 0\Big)\Big).$$
    By \Cref{lem:phantom_exact_aggregate} the aggregate $a = \A^\F_{c_\tau}(P)$ is given as $$a = c_\tau\Big(\frac{1+\tau}{2}, \frac{1-\tau}{2}, 0, \dots, 0\Big) = \Big(\tau, \frac{1-\tau}{2}+\frac{1-\tau}{2(m-1)}, \frac{1-\tau}{2(m-1)}, \dots, \frac{1-\tau}{2(m-1)}\Big).$$ Suppose for contradiction that there exists a moving-phantom mechanism $\A^\F$ such that $\A^\F(P) = a$. Then, $a_j$ corresponds to the median of the votes and the phantoms for alternative $j$ for each $j\in [m]$. Since $a_3$ does not match any of the votes for alternative $3$, some phantom must be at position $\frac{1-\tau}{2(m-1)}$ at a time of normalization.
    But then, for alternative $2$ there are $n-1$ votes at $0$, one vote at $\frac{1-\tau}{2}$ and at least one phantom at $\frac{1-\tau}{2(m-1)} < \frac{1-\tau}{2}$. Thus the median for alternative $2$ is at most $\frac{1-\tau}{2} < \frac{1-\tau}{2}+\frac{1-\tau}{2(m-1)} = a_2$, a contradiction.
\end{proof}

\section{Missing Proof of \Cref{sec:truthful_unanimous_non_phantoms}}

\begin{theorem} \label{thm:existence_truthful_non_and_unanimous_phantoms}
    For $n=2$ and $m=3$, there exists a mechanism $\A$, that is continuous, truthful, neutral, anonymous and unanimous, but is not a moving-phantom mechanism.
\end{theorem}

In this section, we present a moving-phantom mechanism for $m=3$ and $n=2$ that is truthful, anonymous, neutral, continuous, and unanimous. Since proving this result takes significant effort, we divide it into two steps. In a first step, we introduce a novel mechanism that is truthful but not unanimous. In a second step, we adjust this mechanism to make it unanimous and prove its truthfulness by building upon the truthfulness of the first mechanism. 

\subsection{\textsc{VoteCutGreedyMin}: Another Truthful Non-Phantom for $\mathbf{m=3}$} \label{app:votecutgreedy}

We first define a new moving-phantom mechanism \textsc{GreedyMin}, which is defined by the following phantom system: In a first phase, the first phantom moves to 1 while all others remain at 0. In a second phase, all the remaining phantoms move simultaneously from 0 to 1: Formally, we define $\F = \{f_0, \dots, f_{n}\}$, where 
\begin{align*}
    & f_0(t) = \min\{2t,1\} \text{ and} \\
    & f_1(t) = \dots = f_n(t) = \max(0, 2t-1)
\end{align*}

The rule can be interpreted as a dual to the \textsc{GreedyMax} rule, introduced in \Cref{subsec:greedymax} with the following interpretation: In the beginning, we set a budget $b=1$ and a counter $k=0$. Then, we iterate over the alternatives in decreasing order of their minimum votes. Let $j$ be the current alternative. If $\min_{i \in [n]} p_{i,j} > b/(m-k)$, then assign $\min_{i \in [n]} p_{i,j}$ to alternative $j$, decrease $b$ by the same value and increase $k$ by one. Otherwise, assign $b/(m-k)$ to all remaining alternatives and stop the process. The interpretation is that the mechanism greedily assigns the value of the minimum vote to alternatives, unless this would lead to future alternatives receiving more than what is assigned in the current step.

We formalize parts of this intuition in the following lemma. 

\begin{lemma} \label{lem:diag_line_observations}
    Let $P = (p_1, \dots, p_n)$ be a profile and $a = \A^\F(P)$. Then 
    \begin{enumerate}
        \item The sum of medians is normalized at time $t$ with $\frac{1}{2} \le t \le \frac{1+m}{2m}$, \label{itemNormalization}
        \item $a_j \le \max(\frac{1}{m}, \min_{i \in [n]}(p_{i,j}))$ for all $j \in [m]$, \label{lem:diag_line_observations:smaller_than_min_voter}
        \item $a_j \ge \min_{i \in [n]}(p_{i,j})$ for all $j \in [m]$, \label{lem:diag_line_observations:bigger_than_min_voter}
        \item $a_k \ge p_{i,k}$ for all $i \in [n]$, where $k = \argmin_{j \in [m]}(p_{i,j})$. \label{lem:diag_line_observations:bigger_than_min_coordinate}
    \end{enumerate}    
\end{lemma}

\begin{proof}\leavevmode
\begin{enumerate}[label=(\roman*)]
    \item At time $t = \frac{1}{2}$ the first phantom position is 1 and all others are 0, thus the sum of all medians is $\sum_{j \in [m]} \min_{i \in [n]} p_{i,j} \le 1$. At time $t = \frac{1+m}{2m}$ all phantoms have positions at least $\frac{1}{m}$, thus the sum of all medians is at least $m \frac{1}{m} = 1$.
    \item From \Cref{itemNormalization} we know that the lower phantoms never exceed $\frac{1}{m}$. If $\min_{i \in [n]}p_{i,j} \leq \frac{1}{m}$, this implies an upper bound of $\frac{1}{m}$ for $a_j$. Otherwise, the fact that only one phantom is above $\frac{1}{m}$ implies an upper bound of $\min_{i \in [n]} p_{i,j}$ for $a_j$. 
    \item At time $t = \frac{1}{2}$ the median on coordinate $j$ is $\min_{i \in [n]}p_{i,j}$. The claim follows from the monotonicity of the medians over time and the fact that $t\geq \frac{1}{2}$ at time of normalization by \Cref{itemNormalization}.
    \item We fix some $i$ and look at $t = \frac{1+p_{i,k}}{2}$ (the lower phantoms are at position $p_{i,k}$). On coordinate $k$ the median is $p_{i,k}$ and on all other coordinates $j \neq k$ the median is upper bounded by $p_{i,j}$. Thus, the sum of medians is upper-bounded by $1$ showing that we have normalization at time $t \ge \frac{1+p_{i,k}}{2}$ and the median on coordinate $k$ can only increase.
\end{enumerate}
This concludes the proof. 
\end{proof}

\textbf{VoteCutGreedyMin} We make use of the cutoff function $c_\tau$ for the constant function $\tau=0.8$ (defined in section \Cref{subsec:greedymax}), but this time apply it to each of the voters (instead of the aggregate), before applying \textsc{GreedyMin}. We denote the resulting mechanism by $\mathcal{V}$. More precisely, let $P$ be a profile and $\mathcal{A}^{\F}$ the \textsc{GreedyMin} mechanism. Then let $P'$ be the profile where we replace each $p_i$ by $c_{\tau}(p_i)$. Then, let $\mathcal{V}(P) = \mathcal{A}^{\F}(P')$.

\begin{proposition}
    For every $n \in \N$ and any phantom system $\G_n$ there exists some profile $P \in (\Delta^{(2)})^n$ with three alternatives such that $\mathcal{V}(P) \neq \A^{\G_n}(P)$. 
\end{proposition}

\begin{proof}
    For any $n \in \N$, consider the unanimous profile in which every voter reports $(0.9,0.1,0)$. Here, every voter's report is mapped to $c_{\tau}(p_i) = (0.8, 0.15, 0.05)$, which yields the aggregate $\mathcal{V}(P) = (0.8, 0.15, 0.05)$. Suppose $\mathcal{V}$ was a moving-phantom mechanism. Then, at time of normalization, the phantom $g_n$ has to be at $0.05$ (due to alternative $3$). However, this implies that the median on alternative $2$ is at most $0.1$, a contradiction. 
\end{proof}

The following lemma will be helpful throughout this section and states that in the case of three alternatives, in order to decrease the $\ell_1$-distance of two normalized vector, one vector has to ``move towards'' the other on at least two coordinates.

\begin{lemma} \label{lem:twocoordinatesmove}
    Let $a,a^*,v \in \Delta^{(2)}$. If $a_i^* \geq a_i\geq v_i$ or $a_i^* \leq a_i\leq v_i$ for at least two $i \in \{1,2,3\}$, then $$\ellone{a}{v} \leq \ellone{a^*}{v}.$$
\end{lemma}

\begin{proof}
    Assume that the condition holds for $i=1$ and $i=2$. 
    \begin{CaseTree}
        \item $a_1 \geq a^*_1$ and $a_2 \geq a^*_2$ or $a_1 < a^*_1$ and $a_2 < a^*_2$. \\
        Then, $|a_1 - a^*_1| + |a_2 - a^*_2| = |a_3 - a^*_3|$ and thus $$\ellone{a^*}{v} - \ellone{a}{v} \geq |a_1 - a^*_1| + |a_2 - a^*_2| - |a_3 - a^*_3| = 0 .$$
        \item $a_1 \geq a^\star_1$ and $a_2 < a^\star_2$ or $a_1 < a^\star_1$ and $a_2 \geq a^\star_2$. \\
        Assume wlog that $a_1 \geq a^\star_1$. Also, assume that $a_3 \geq a^\star_3$ (the case when $a_3 < a^\star_3$ follows analogously.) 
        Then, $|a_3 - a^*_3| \leq |a_2 - a^*_2|$ and thus $$\ellone{a^*}{v} - \ellone{a}{v} \geq |a_1 - a^*_1| + |a_2 - a^*_2| - |a_3 - a^*_3| \geq 0 .$$
    \end{CaseTree}
    This concludes the proof. 
\end{proof}

\begin{proposition} 
\label{prop:VoteCutGreedyMin}
    For $m=3$ and any $n \in \N$, the mechanism $\mathcal{V}$ is truthful. 
\end{proposition}

\begin{proof}
    It follows directly from the definition of the mechanism that it is anonymous and neutral. Now, assume there is a truthfulness violation for some voter. By anonymity we can assume that this voter is voter $1$, i.e., there is a profile $P = (p_1, \dots, p_n)$ with aggregate $a=\V(P)$ and a profile $P^* = (p_1^*, p_2, \dots, p_n)$ that yields voter 1 a preferred result $a^* = \V(P)$. Let $p_1' = c_\tau(p_1)$. We claim that we can assume the following without loss of generality:
    \begin{enumerate}[label=(\roman*)]
        \item $p_{1,1} \ge p_{1,2} \ge p_{1,3}$
        \item $p_{1,1} > 0.8$
        \item $p_{i,j} \le 0.8$ for all $2 \le i \le n$ and $j \in [3]$ \label{nomapping}
        \item $p^*_{1,j} \le 0.8$ for all $j \in [3]$ \label{nomapping2}
    \end{enumerate}

    To see why, first, note that if $\max_{j \in [3]}p_{1,j} \leq 0.8$, then the truthfulness follows directly from the truthfulness of \textsc{GreedyMin} for the profile $P'$. By neutrality we can fix the order of the three alternatives with respect to voter $1$ arbitrarily. Moreover, for the first voter it is irrelevant whether the remaining voters vote like in $P$ and then get mapped to their corresponding votes in $P'$ or whether they directly vote like in $P'$, thus, we can assume \ref{nomapping}. Due to a similar argument, we can assume in \ref{nomapping2} that voter $1$ does not misreport a $0.8$-extreme vote for any alternative but instead reports a vote from the image of the mapping.

    Since $p_{1,1} > 0.8$, we have $p_{1,1}' = 0.8$ and $p_{1,2}' - p_{1,2} = p_{1,3}' - p_{1,3} > 0$. We also know that $p_{1,2}', p_{1,3}' < 0.2$.
    
    We first show that $a_1^* \le a_1$:
    
    \begin{CaseTree}
        \item $\min_{i \in [n]} p_{i,1} \le \frac{1}{3}$.
            Then on each coordinate there is a vote lower or equal than $\frac{1}{3}$ and thus by \cref{lem:diag_line_observations:smaller_than_min_voter}, $a_j \le \frac{1}{3}$ for all $j \in [m]$ and therefore $a = (\frac{1}{3}, \frac{1}{3}, \frac{1}{3})$. Also by \cref{lem:diag_line_observations:smaller_than_min_voter}, we get that $a_1^* \le \frac{1}{3} = a_1$. This is because $\argmin_{i \in [n]} p_{i,1} \neq 1$ (since $p_{1,1}>0.8$) and therefore $\min\{p^*_{1,1}, p_{2,1}, p_{3,1}\} \le \frac{1}{3}$.
        
        \item $\min_{i \in [n]} p_{i,1} > \frac{1}{3}$.
            Since the lower phantoms can only move up to $\frac{1}{3}$ before reaching normalization by \Cref{itemNormalization}, we have $a_1 = \min_{i \in [n]} p_{i,1}$. By \cref{lem:diag_line_observations:smaller_than_min_voter} we know that $a_1^* \le \max \{ \frac{1}{m}, \min\{p^*_{1,1}, p_{2,1}, p_{3,1}\} \} \le \max \{ \frac{1}{m}, \min_{i \in [n]} p_{i,1} \}  = \min_{i \in [n]} p_{i,1} = a_1$.\\
            
    \end{CaseTree}

    {

    Hence, we know that \begin{equation}
        a^*_1 \leq a_1 \leq p'_{1,1} \leq p_{1,1}. \label{firstproject}
    \end{equation}
     From \Cref{lem:diag_line_observations:bigger_than_min_coordinate}, we also know that
    \begin{equation}
        p_{1,3} < p_{1,3}' \leq a_3. \label{thirsproject}
    \end{equation}
    Moreover, we assume for contradiction that $\ellone{a^*}{p_1} - \ellone{a}{p_1} < 0$, which we claim to imply 
    \begin{equation}
        a_2 < p_{1,2} < p'_{1,2}. \label{secondproject}
    \end{equation}
    Namely, by \Cref{lem:twocoordinatesmove}, $a^*$ has to move ``towards'' voter $1$ on at least two coordinates. Moreover, since $a^*$ moves downwards on alternative $1$, i.e., $a^*_1 \leq a_1$, there has to be a beneficial upwards move on alternative $2$ or $3$. Since an upwards move on alternative $3$ is not beneficial for voter $1$ by \Cref{thirsproject}, there has to be a beneficial upwards movement on alternative $2$. Thus, $a_2 < p_2 < p'_2$.

    Now, consider 
    \begin{align*}
        &\frac{1}{2}\Big(\ellone{a^*}{p_1} - \ellone{a}{p_1} - \big( \ellone{a^*}{p'_1} - \ellone{a}{p'_1}\big)\Big) \\ 
        & = \sum_{j \in [3]} \Big(\min(a_j,p_{1,j}) - \min(a_j,p'_{1,j}) + \min(a^*_j,p'_{1,j}) - \min(a^*_j,p_{1,j}) \Big) \tag{\Cref{footnoteLemma}}\\ 
        & = \min(a^*_2,p'_{1,2}) - \min(a^*_2,p_{1,2}) + p_{1,3} - p'_{1,3} + \min(a^*_3,p'_{1,3}) - \min(a^*_3,p_{1,3}) \tag{\Cref{firstproject,secondproject,thirsproject}}\\ 
        \intertext{Case 1: $a_2^* \leq p'_{1,2}$. Then, $a^*_3\geq p'_{1,3}> p_{1,3}$ by normalization of $a^*$ and $p'_1$ and \Cref{firstproject}.}
        & = \min(a^*_2,p'_{1,2}) - \min(a_2^*,p_{1,2})) \geq 0
        \intertext{Case 2: $a_2^* > p'_{1,2}$.}
        & = p'_{1,2} - p_{1,2} + p_{1,3} - p'_{1,3} + \min(a^*_3,p'_{1,3}) - \min(a^*_3,p_{1,3}) \\ 
        & = \min(a^*_3,p'_{1,3}) - \min(a^*_3,p_{1,3}) \geq 0. 
    \end{align*}

    Since $\big(\ellone{a^*}{p'_1} - \ellone{a}{p'_1}\big) \geq 0$ this yields a contradiction to our assumption that $\big(\ellone{a^*}{p_1} - \ellone{a}{p_1}\big) < 0$.}
\end{proof}

\subsection{A Truthful and Unanimous Non-Phantom Mechanism} \label{app:unamvotecutgreedy}
The mechanism \textsc{VoteCutGreedyMax} is notably not unanimous since all voters reporting more than $0.8$ for some alternative will lead to this alternative receiving only $0.8$. In the following, we show how we can adapt the cutoff function to achieve unanimity while maintaining truthfulness. The key idea is to apply the cutoff function for some $0.8$-extreme voter only in situations where at least one other voters reports a low value (namely, below $0.7$) on that alternative. If the other voter votes above 0.8, we do not apply the cutoff. If the voter votes between 0.7 and 0.8, we linearly increase the factor of the cutoff from 0 to 1.

We define the new cutoff function $\Tilde{c}$ formally as follows. We assume from now on that $n=2$. Thus, we only have two voters and define $P = (v, w)$. 

Let $k \in \argmax_{j \in [m]}(v_j)$. Let $\gamma(v,w) = (v_1 - 0.8) \max(0, \min(\frac{0.8-w_1}{0.8-0.7}, 1))$. 
Then  
\[
    \Tilde{c}(v,w)_j = 
    \begin{cases} 
        v_j & \text{if } v_k \le 0.8, \\
        v_j-\gamma(v,w) & \text{if } v_k > 0.8 \text{ and } j = k, \\
        v_j + \frac{\gamma(v,w)}{2} & \text{if } v_k > 0.8 \text{ and } j \neq k. \\
    \end{cases}
\]

We define the mechanism \textsc{UnanimousVoteCutGreedyMin} as $\mathcal{U}(P) = \A^{\F}(P')$, where $\F$ is the phantom system of \textsc{GreedyMin} and $P'=(\Tilde{c}(v,w),\Tilde{c}(w,v))$.

It follows directly from the definition that $\mathcal{U}$ is unanimous. Namely, consider a profile $P=(v,v)$. Then, $P'=(\Tilde{c}(v,w), \Tilde{c}(w,v)) = (v,v) = P$, moreover, $\A^{\F}(v,v) = v$. 

Before showing that $\mathcal{U}$ is in fact truthful and not a moving-phantom mechanism, we will show several helpful Lemmas.

\begin{lemma} \label{lem:moving_other_voter_away_plus}
    For $v, w,w' \in \Delta^{(2)}$ with $w' = (w_1+\delta, w_2-\frac{\delta}{2}, w_3-\frac{\delta}{2})$ for some $\delta > 0$, if $w_1 \ge 0.8$ and $w_1 \ge v_1$ this implies $$\ellone{v}{\A^\F(v,w)} \le \ellone{v}{\A^\F(v,w')}.$$
\end{lemma}

\begin{proof}
    Let $a = \A^\F(v,w)$ and $a' = \A^\F(v,w')$.
    
    \begin{CaseTree}
        \item $v_1 \le \frac{1}{3}$\\
        Then the minima on all coordinates are below $\frac{1}{3}$ for both profiles and thus by \Cref{lem:diag_line_observations:smaller_than_min_voter} $a = a' = (\frac{1}{3}, \frac{1}{3}, \frac{1}{3})$.
        \item $v_1 > \frac{1}{3}$\\
        Then by \Cref{lem:diag_line_observations:smaller_than_min_voter,lem:diag_line_observations:bigger_than_min_voter}, we get that $a_1 = v_1 = a_1'$. Without loss of generality, assume $v_2 \ge v_3$. Then $a_3, a_3' \ge v_3$ by \Cref{lem:diag_line_observations:bigger_than_min_coordinate}.
        \begin{CaseTree}
            \item $v_2 \le w_2$\\
            Then for the profile $(v,w)$ normalization is reached, when the lower phantoms reach position $v_3$, since the medians on the second and third coordinate evaluate to $\med(v_2, w_2, 1, v_3, v_3) = v_2$, $\med(v_3, w_3, 1, v_3, v_3) = v_3$ and $v_1 + v_2 + v_3 = 1$. Then $a = v$ and the $\ell_1$-distance to $a'$ cannot be smaller than 0. 
            
            \item $v_2 > w_2$
            \begin{CaseTree}
                \item $w_2 > \frac{1-v_1}{2}$\\
                Then, for the profile $(v,w)$ normalization is reached, when the lower phantoms reach position $1-v_1-w_2$. Hence, $a = (v_1, w_2, 1-v_1 - w_2)$. We claim that $a'_2 \leq w_2 = a_2$. To see why, note that $a'_2 = \max\{\frac{1-v_1}{2},w'_2\} \leq \max\{\frac{1-v_1}{2},w_2\} = w_2$. However, we know by \Cref{lem:twocoordinatesmove} and $a'_1=a_1$ that, in order for $\ellone{v}{a'} < \ellone{v}{a}$, $a'$ has to move towards $v$ (compared to $a$) for at least two alternatives. Thus, it has to hold that $a'_2 > a_2$, a contradiction. 
                \item $w_2 \leq \frac{1-v_1}{2}$\\
                Then, for the profile $(v,w)$ normalization is reached, when the lower phantoms reach position $\frac{1-v_1}{2}$ and $a = (v_1, \frac{1-v_1}{2}, \frac{1-v_1}{2})$. This is because $\min\{v_2,w_2\} \leq \frac{1-v_1}{2}$ and $\min\{v_3,w_3\} \leq \frac{1-v_1}{2}$. Note that these condition also holds for $w'$, and thus $a'=a$.
            \end{CaseTree}
        \end{CaseTree}
    \end{CaseTree}
    This concludes the proof. 
\end{proof}

\begin{lemma} \label{lem:moving_other_voter_away_minus}
    For $v, w,w' \in \Delta^{(2)}$ with $w' = (w_1-\delta, w_2+\frac{\delta}{2}, w_3+\frac{\delta}{2})$ for some $\delta \geq 0$, if $w_1 \ge 0.8$ and $w_1 \le v_1$, this implies $$\ellone{v}{\A^\F(v,w)} \le \ellone{v}{\A^\F(v,w')}.$$
\end{lemma}

\begin{proof}
Let $a = \A^\F(v,w)$ and $a' = \A^\F(v,w')$. Note that $a_1 = w_1$. 

    First consider the case when $w'_1 \leq \frac{1}{3}$. Then, $a'_1 \leq \frac{1}{3}$ and 
    $$\ellone{v}{a'} \geq |v_1 - a'_1| \geq |v_1 - a_1| + 0.8 - \frac{1}{3} $$ 
    and $$\ellone{v}{a} = |v_1 - a_1| + |v_2 - a_2| + |v_3 - a_3| \leq |v_1 - a_1| + 0.4,$$ 
    where the last inequality follows from the fact that $v_2 + v_3 \leq 0.2$ and $a_2 + a_3 \leq 0.2$.

    Hence, we can assume that $w'_1 \geq \frac{1}{3}$ and therefore $a'_1 = w'_1$. We claim that $a'_j \geq a_j$ for $j \in \{2,3\}$. Then, since $v_j \leq a_j \leq a'_j$ for some $j \in \{2,3\}$ and $v_1 \geq a_1 > a'_1$, the lemma statement follows by \Cref{lem:twocoordinatesmove}. 

    To prove the claim, we can consider the following cases wlog.

    \begin{CaseTree}
        \item $\min(v_2,w_2) \geq \min(v_3,w_3)$ and $\min(v_2,w_2 + \frac{\delta}{2}) \leq \min(v_3,w_3+ \frac{\delta}{2})$

        This implies that $\min(v_3,w_3) \leq a_3 \leq \frac{1-w_1}{2}$ and $\min(v_2,w_2+\frac{\delta}{2}) \leq a'_2 \leq \frac{1-w_1}{2} +\frac{\delta}{2}$. Thus, $$a'_3 = \max\Big(\min(v_3,w_3 + \frac{\delta}{2}),\frac{1-w_1}{2} + \frac{\delta}{2}\Big) = \frac{1-w_1}{2} + \frac{\delta}{2} > \frac{1-w_1}{2} \geq a_3.$$

        Due to normalization of $a$ and $a_1=w_1$ this implies that $$a'_2 = \frac{1-w_1}{2} + \frac{\delta}{2} \geq \max\Big(\min(v_2,w_2 + \frac{\delta}{2}), \frac{1-w_1}{2} + \frac{\delta}{2}\Big) \geq \max\Big(\min(v_2,w_2), \frac{1-w_1}{2}\Big) =a_2.$$

        \item $\min(v_2,w_2) \geq \min(v_3,w_3)$ and $\min(v_2,w_2 + \frac{\delta}{2}) \geq \min(v_3,w_3+ \frac{\delta}{2})$

        Then, 
        \begin{align*}
        a_2 &= \max\big(\min(v_2,w_2),\frac{1-w_1}{2} \big) \leq  \max\big(\min(v_2,w_2 + \frac{\delta}{2}),\frac{1-w_1}{2} + \frac{\delta}{2}\big)
        = a'_2 \\ &\leq \max\big(\min(v_2,w_2),\frac{1-w_1}{2} \big) + \frac{\delta}{2} = a_2 + \frac{\delta}{2}.
        \end{align*}
        By normalization of $a$ and $a_1 = w_1$ this implies that $a'_3 \in [a_3 + \frac{\delta}{2},a_3 + \delta]$. 
    \end{CaseTree}
    As discussed above, we can apply \Cref{lem:twocoordinatesmove} to conclude the proof of the lemma. 
\end{proof}

For the rest of this section, we assume without loss of generality that the first voter's vote is ordered decreasingly, i.e. $v_1 \ge v_2 \ge v_3$.

\begin{proof}[Proof of \Cref{thm:existence_truthful_non_and_unanimous_phantoms}]

We start by proving that the mechanism $\mathcal{U}$ is not a moving-phantom mechanism, continuous, anonymous, and neutral. 

\textbf{Non-phantom:} In the following, we show that $\mathcal{U}$ is not a moving-phantom mechanism.
Consider the profiles $P_1 = ((0.84, 0.16, 0.0), (0.7, 0.3, 0.0))$ and $P_2 = ((0.82, 0.18, 0.0), (0.7, 0.3, 0.0))$. Applying the cutoff function $\Tilde{c}$ to the profiles results in the profiles $P_1' = ((0.80, 0.18, 0.02), (0.7, 0.3, 0.0))$ and $P_2' = ((0.80, 0.19, 0.1), (0.7, 0.3, 0.0))$. Running \textsc{GreedyMin} on these profiles, gives us the aggregate of \textsc{UnanimousVoteCutGreedyMin} as $\mathcal{U}(P_1) = (0.7, 0.18, 0.12)$ and $\mathcal{U}(P_1) = (0.7, 0.19, 0.11)$.

Now, suppose there is a moving-phantom mechanism $\A^\F$ with the phantom system $F_2 = \{f_0,f_1,f_2\}$ with $\A^\F(P_1) = \mathcal{U}(P_1)$ and $\A^\F(P_2) = \mathcal{U}(P_2)$. At the time of normalization for both profiles, one phantom must be at a position higher or equal 0.7 for the median to be 0.7. Since in $P_1$ no vote is at positions 0.18 on the second alternative or at position 0.12 on the third alternative, we must have $f_1(t_1) = 0.18$ and $f_2(t_1) = 0.12$ at the time of normalization $t_1$ for profile $P_1$. Similarly, for profile $P_2$, we must have $f_1(t_2) = 0.19$ and $f_2(t_2) = 0.11$ at the time of normalization $t_2$. This contradicts the monotonicity of $f_1$ and $f_2$.

\textbf{Anonymity:} Follows directly from the anonymity of the cutoff functiom $\Tilde{c}$ and the anonymity of \textsc{GreedyMin}. 

\textbf{Neutrality:} Follows directly from the neutrality of the cutoff function $\Tilde{c}$ and the neutrality of \textsc{GreedyMin}. 

\textbf{Continuity:} One can easily verify that $\Tilde{c}$ is a continuous function and \textsc{GreedyMin} is continuous du to being a moving-phantom mechanism. Since \textsc{UnanimousVoteCutGreedyMin} is a composition of these two functions, continuity follows.

\textbf{Truthfulness:} Now, we turn to proving the truthfulness of the mechanism. To this end, we follow a case distinction on the votes $v$ and $w$, where $v$ takes the role of the manipulating voter. That is, suppose there is a vote $v^\star$, such that $\ellone{v}{\mathcal{U}(v,w)} > \ellone{v}{\mathcal{U}(v^\star,w)}$. 

We define $a = \mathcal{U}(v,w)$ and $a^\star = \mathcal{U}(v^\star,w)$. Recall that $a = \mathcal{U}(v,w) =\mathcal{A}^{\F}(\Tilde{c}(v,w),\Tilde{c}(w,v))$ and similarly $a^* = \mathcal{U}(v^*,w) = \mathcal{A}^{\F}(\Tilde{c}(v^*,w),\Tilde{c}(w,v^*))$. 

    \begin{CaseTree}
        \item $\max(w) < 0.8$
        \begin{CaseTree}
            \item $v_1 < 0.8$\\
            Then, no cutoff is applied, that is, $\Tilde{c}(v,w)=v$ and $\Tilde{c}(w,v) = w$ since all reports are below $0.8$. Thus, $a = \mathcal{U}(v,w) = \mathcal{A}^{\F}(v,w)$ and $a^\star=\mathcal{U}(v^\star,w) = \mathcal{A}^{\F}(v^\star,w)$. Hence, if $\ellone{a^\star}{v} < \ellone{a}{v}$, this would violate the truthfulness of \textsc{GreedyMin}, which follows from \textsc{GreedyMin} being a moving-phantom mechanism. 
            \item $v_1 \ge 0.8$
            \begin{CaseTree}
                \item $\max(w) \le 0.7$\\
                Then, the cutoff function $\Tilde{c}$ and $c_{\tau}$ for $\tau$ being the constant function $0.8$ coincide for this instance. Namely, $\Tilde{c}(w,v)=c_{\tau}(w)=w$ and $\Tilde{c}(v,w) = c_{\tau}(v)$. Thus, $a = \mathcal{V}(v,w) = \mathcal{A^{\F}}(c_\tau(v),c_\tau(v))$ and $a^\star=\mathcal{V}(c_\tau(v^\star),c_\tau(w))$. In particular, if $\ellone{a^*}{v}<\ellone{a}{v}$ this would violate the truthfulness of \textsc{VoteCutGreedyMin}, established in \Cref{prop:VoteCutGreedyMin}.

                \item $\max(w) > 0.7$
                \begin{CaseTree}
                    \item $w_1 \le 0.7$ \label{referenceCase} \\ This implies that $\min(v_i,w_i)\leq \frac{1}{3}$ for all $i \in [3]$ and therefore by \Cref{lem:diag_line_observations} that $a=(\frac{1}{3},\frac{1}{3},\frac{1}{3})$. We claim that $\ellone{v}{a^\star}<\ellone{v}{a}$ implies that $a^\star_1>\frac{1}{3}$. Since $v_1 \geq 0.8 \geq a_1$ this implies $a_j > 0.2 \geq v_j$ for both $j \in \{2,3\}$. Now, by \Cref{lem:twocoordinatesmove} we know that two alternatives have to ``move towards'' $v$. If alternative $1$ is among them, then $a^\star_1 > a_1$ follows. Otherwise, $a^\star_2<a_2$ and $a^\star_3<a_3$, which by normalization implies $a^\star_1 > a_1$. 
                    However, since $\min\{v^\star_1,w_1\} \leq 0.3 < \frac{1}{3}$ it follows from \Cref{lem:diag_line_observations:smaller_than_min_voter} that $a^\star_1 \leq \frac{1}{3}$. 
                    \item $w_1 > 0.7$\\
                    It will be sufficient to show that $\ellone{v}{\mathcal{U}(v,w)} = \ellone{v}{\A^\F(v,w)}$ for all $v \in \Delta^{(2)}$. Once we have established this, we can show truthfulness of $\mathcal{U}$ by contradiction. If a successful manipulation $v^\star$ exists, then $\ellone{v}{\A^\F(\Tilde{c}(v^\star,w),w)} = \ellone{v}{\mathcal{U}(v^\star,w)} < \ellone{v}{\mathcal{U}(v,w)} = \ellone{v}{\A^\F(v,w)}$, contradicting truthfulness of \textsc{GreedyMin}. Note that the first equality holds because $\Tilde{c}(w)=w$, since $\max (w)<0.8$. The remainder of the proof of this case is devoted to proving the sufficient condition. Clearly it holds for $v_1<0.8$, since in that case $\mathcal{U}(v,w)=\A^\F(v,w)$, so we focus on $v_1 \ge 0.8>w_1$, we extensively use the fact that $v_2+v_3 \le \Tilde{c}(v_2)+\Tilde{c}(v_3) = 1-\Tilde{c}(v_1) \le 0.2 < 1-w_1$. We treat several subcases:

                    \begin{CaseTree}

                    \item Suppose that $\frac{1-w_1}{2} \ge \Tilde{c}(v,w)_2 \ge \Tilde{c}(v,w)_3$. It is therefore also the case that $\frac{1-w_1}{2} \ge v_2 \ge v_3$. In both cases, normalization is reached when phantoms $f_1$ and $f_2$ reach $\frac{1-w_1}{2}$, and therefore $\mathcal{U}(v,w) = \A^\F(v,w) = (w_1, \frac{1-w_1}{2}, \frac{1-w_1}{2})$. 

                    \item Suppose that $v_2 \ge \frac{1-w_1}{2}$ and $w_2 \ge \frac{1-w_1}{2}$. This implies that $\Tilde{c}(v,w)_2 \ge \frac{1-w_1}{2}$, and therefore $\Tilde{c}(v,w)_3$, $v_3$, and $w_3$ are all less than $\frac{1-w_1}{2}$. Then we have $\A^\F(v,w)=(w_1, \min (v_2, w_2), x)$, for some $x$ with $\frac{1-w_1}{2} > x > v_3$, with normalization being reached when $f_1=f_2=x$. (Note that if $x \le v_3$ then the sum of entries of $v$ is more than the sum of entries of the aggregate, violating normalization.) By the same logic, $\mathcal{U}(v,w) = (w_1, \min (\Tilde{c}(v,w)_2, w_2), y)$ for some $y$ with $\frac{1-w_1}{2} > y > \Tilde{c}(v,w)_3 \ge v_3$. Note that $y \le x$ since $\Tilde{c}(v,w)_2 \ge v_2$, and $y<x$ only if $\min (v_2, w_2)< \min (\Tilde{c}(v,w)_2, w_2)$, which requires $v_2 < w_2$. In this case, $\mathcal{U}(v,w)_3$ is $(x-y)$ \emph{closer} to $v_3$ than $\A^\F(v,w)_3$ is, but $\mathcal{U}(v,w)_2$ is $(x-y)$ \emph{further} from $v_2$ than $\A^\F(v,w)_2$ is. Therefore, $\ellone{v}{\mathcal{U}(v,w)} = \ellone{v}{\A^\F(v,w)}$.

                    \item Suppose that $v_2 \ge \frac{1-w_1}{2}$ and $w_2 \le \frac{1-w_1}{2}$ (which implies that $v_3 < \frac{1-w_1}{2}$ and $w_3 \ge \frac{1-w_1}{2}$). Then it is also the case that $\Tilde{c}(v,w)_2 > \frac{1-w_1}{2}$ and $\Tilde{c}(v,w)_3 < \frac{1-w_1}{2}$. In both cases, normalization is reached when phantoms $f_1$ and $f_2$ reach $\frac{1-w_1}{2}$, and therefore $\mathcal{U}(v,w) = \A^\F(v,w) = (w_1, \frac{1-w_1}{2}, \frac{1-w_1}{2})$. 

                    \item Suppose that $\Tilde{c}(v,w)_2 > \frac{1-w_1}{2} > v_2$. This implies that $\frac{1-w_1}{2} > \Tilde{c}(v,w)_3 > v_3$. There are two subcases:

                    If $w_2 < \frac{1-w_1}{2}$ then, again, normalization is reached when phantoms $f_1$ and $f_2$ reach $\frac{1-w_1}{2}$, and therefore $\mathcal{U}(v,w) = \A^\F(v,w) = (w_1, \frac{1-w_1}{2}, \frac{1-w_1}{2})$.

                    If $w_2 > \frac{1-w_1}{2}$ then it is still the case that $\A^\F(v,w)=(w_1, \frac{1-w_1}{2}, \frac{1-w_1}{2})$. However, $\mathcal{U}(v,w)=(w_1, \min (\Tilde{c}(v,w)_2, w_2), z)$ for some $z$ with $\frac{1-w_1}{2} > z > \Tilde{c}(v,w)_3 > v_3$. (If $z \le \Tilde{c}(v,w)_3$ then coordinates of $\Tilde{c}(v,w)$ sum to strictly more than those of the aggregate, violating normalization.) 
                    In this case, $\mathcal{U}(v,w)_3$ is $(\frac{1-w_1}{2}-z)$ \emph{closer} to $v_3$ than $\A^\F(v,w)_3$ is, but $\mathcal{U}(v,w)_2$ is $(\frac{1-w_1}{2}-z)$ \emph{further} from $v_2$ than $\A^\F(v,w)_2$ is. Therefore, $\ellone{v}{\mathcal{U}(v,w)} = \ellone{v}{\A^\F(v,w)}$.
                \end{CaseTree}
                \end{CaseTree}
            \end{CaseTree}
        \end{CaseTree}
        \item $\max(w) \ge 0.8$, but $w_1 < 0.8$
        \begin{CaseTree}
            \item $v_1 > 0.8$ \\
            These are almost the same case assumptions like \ref{referenceCase}, with the only difference that $\max(w)>0.8$ instead of $\max(w) \in [0.7,0.8]$. Since \ref{referenceCase} only uses the fact that $\max(w) \geq 0.7$, we can apply the exact same argumentation. 
            
            \item \label{referenceCase2} $v_1 \leq 0.8$\\ 
            Wlog assume that $\argmax_{j \in [3]}(w) =2$. By the $v_1 \geq v_2 \geq v_3$, we know that $v_2 \leq \frac{1}{3}$ and therefore $\Tilde{c}(w,v) = c_\tau(w)$ for $\tau$ being the constant $0.8$ and $\Tilde{c}(v,w) = v$. Now, consider the profile $P' = (v, \Tilde{c}(w,v^\star))$, where we update the vote $w$ with the cutoff function as if $v$ has moved to $v^\star$, but we do not actually move $v$ yet. Note that $v$, $\Tilde{c}(w,v)$ and $\Tilde{c}(w,v^\star)$ satisfy the conditions of the three vectors in \Cref{lem:moving_other_voter_away_plus}. Hence, $\ellone{v}{\A^\F(v,\Tilde{c}(w,v))} \leq \ellone{v}{\A^\F(v,\Tilde{c}(w,v^\star))}$.
            Moreover, we can apply the truthfulness of \textsc{GreedyMin} and obtain $\ellone{v}{\A^\F(v,\Tilde{c}(w,v^*))} \leq \ellone{v}{\A^\F(v^*,\Tilde{c}(w,v^*))}$, which proves this case. 
        \end{CaseTree}
        
        \item $w_1 \ge 0.8$
        \begin{CaseTree}
            \item $v_1 \le 0.7$\\
            This case works similar to \ref{referenceCase2}. Since $0.7 \geq v_1 \geq \frac{1}{3}$ we know that $\Tilde{c}(v,w) = v$ and $\Tilde{c}(w,v) = c_\tau(w)$ for $\tau$ being the constant $0.8$. Now, we can follow the exact same argumentation as in \ref{referenceCase2}.
            \item $v_1 > 0.7$
            \begin{CaseTree}
                \item $w_1 \ge 0.9$ \\
                We show that $\mathcal{U}(v^\star,w) = \A^\F(v^\star,w)$ for all $v^\star \in \Delta^{(2)}$ and apply the truthfulness of $\A^\F$.

                If $v^\star_1 \ge 0.8$, then no cutoff is applied and the claim follows. If $v^\star_1 \le 0.7$, then the cutoff is fully applied, i.e., $\Tilde{c}(w,v^\star) = c_{\tau}(w)$ for constant $\tau$ being $0.8$. Since $w_1 \ge 0.9$, then $\Tilde{c}(w,v^\star)_j \le 0.15$ for $j \in \{ 2,3 \}$. Normalization is reached after the lower phantoms cross $0.15$, and the aggregate is $(v^\star_1, \frac{1-v^\star_1}{2}, \frac{1-v^\star_1}{2})$ (if $v^\star_1 \ge \frac{1}{3}$) or $(\frac{1}{3},\frac{1}{3},\frac{1}{3})$ (if $v^\star_1 < \frac{1}{3}$). Note that the same argument applies to the instance where we do not cut off $w$. Thus, $\mathcal{U}(P) = \A^\F(P)$.

                If $0.7 < v^\star_1 <0.8$ then write $v^\star_1=0.8-x$ for $0<x<0.1$. We have $\Tilde{c}(w,v^\star)_j \le 1-w_1+\frac{x}{0.2}(w_1-0.8) =1-(1-5x)w_1-4x \le 1-(1-5x) \cdot 0.9-4x = 0.1+\frac{x}{2}$ for $j \in \{ 2,3 \}$. Therefore normalization is achieved when the lower phantoms reach $0.1+\frac{x}{2}$ and the aggregate is $(v^{\star}_1, 0.1+\frac{x}{2}, 0.1+\frac{x}{2})$. The same argument applies to the non-cut off version since $w_j < \Tilde{c}(w,v^\star)_j$ for $j \in \{ 2,3 \}$. 
                
                \item $w_1 < 0.9$
                \begin{CaseTree}
                    \item $w_1 < v_1$
                    Since $w_1 \ge 0.8$, there is no cutoff applied for the profile $(v,w)$.
                    \begin{CaseTree}
                        \item $v_1^\star > v_1$\\
                        Then, $\Tilde{c}(v,w) = v,\Tilde{c}(v^\star,w) =v^\star,\Tilde{c}(w,v) = w$, and $\Tilde{c}(w,v^\star) = w$. Thus, $\mathcal{U}(P) =\A^\F(P)$ and $\mathcal{U}(P^\star) = \A^\F(P^\star)$. Thus, $\ellone{v}{a} \leq \ellone{v}{a^\star}$ follows from the truthfulness of \textsc{GreedyMin}. 
                        \item $v_1^\star \le v_1$ 
                        Consider the profile $(v, \Tilde{c}(w,v^\star))$, where we update the vote $w$ with the cutoff function as if $v$ has moved to $v^\star$, but we do not actually move $v$ yet. Now, note that $v,w$, and $\Tilde{c}(w,v^\star)$ satisfy the conditions of the three vectors in \Cref{lem:moving_other_voter_away_minus} and hence $\ellone{v}{\A^F(v,w)} \leq \ellone{v}{\A^\F(v,\Tilde{c}(w,v^\star))}$. Now, by the truthfulness of \textsc{GreedyMin} we get $\ellone{v}{\A^\F(v,\Tilde{c}(w,v^\star))} \leq \ellone{v}{\A^\F(\Tilde{c}(v^\star,w),\Tilde{c}(w,v^\star))}$ which proves the claim. 
                    \end{CaseTree}
                    
                    \item $w_1 \ge v_1$
                    \begin{CaseTree}
                        \item $v_1^\star > v_1$\\
                        Note that $\Tilde{c}(v,w)=v$ and $\Tilde{c}(v^\star,w) = v^\star$. This case is again similar to \ref{referenceCase2}. That is, consider the profile $(v, \Tilde{c}(w,v^\star))$, where we update the vote $w$ with the cutoff function as if $v$ has moved to $v^\star$, but we do not actually move $v$ yet. Then, the vectors $v,\Tilde{c}(w,v)$, and $\Tilde{c}(w,v^\star)$ satisfy the conditions of \Cref{lem:moving_other_voter_away_plus} and hence $\ellone{v}{\A^\F(v,\Tilde{c}(w,v))} \leq \ellone{v}{\A^\F(v,\Tilde{c}(w,v^\star))}$. Moreover, by the truthfulness of \textsc{GreedyMin} we know that $\ellone{v}{\A^\F(v,\Tilde{c}(w,v^\star))} \leq \ellone{v}{\A^\F(v^\star,\Tilde{c}(w,v^\star))}$, which proves the claim. 

                        \item $v_1^\star \le v_1$\\
                        We summarize the case assumptions: $w_1 \in [0.8,0.9)$, $v_1 \in (0.7,w_1]$, and $v^*_1 < v_1$. By \Cref{lem:diag_line_observations:bigger_than_min_voter,lem:diag_line_observations:smaller_than_min_voter} it follows directly that $a_1=v_1$.

                        We claim that if $a^*_1 \neq v_1^*$ holds truthfulness follows directly. To see why, note by \Cref{lem:diag_line_observations:bigger_than_min_voter,lem:diag_line_observations:smaller_than_min_voter} $a^*_1 = v^*_1$ unless $v^*_1 < \frac{1}{3}$. On the other hand, if $v_1^* < \frac{1}{3}$, then the outcome is $a^* = (\frac{1}{3},\frac{1}{3},\frac{1}{3})$. In this case, $$\ellone{a^*}{v} \geq v_1 - \frac{1}{3} + \frac{2}{3} - (1-v_1) \geq 0.4 + \frac{1}{3}.$$
                        However, $\ellone{a}{v} \leq 2 \cdot 0.3$ since $a_1 = v_1 \geq 0.7$. 

                        \bigskip 

                        Hence, we can assume $a^*_1 = v_1^*$ from now on. We define $\Tilde{w} = \Tilde{c}(w,v)$ and note that $\Tilde{c}(v,w)=v$. We claim that $a_2 \geq a_3$. This is because $\min\{v_3,\Tilde{w}_3\} \leq \frac{1-v_1}{2}$ by the order of $v$, and hence $a_3 \leq \max \{\frac{1-v_1}{2},\min\{v_3,\Tilde{w}_3\}\} = \frac{1-v_1}{2}$, which, by normalization of $a$ and $a_1=v_1$ implies that $a_3 \leq a_2$. We claim that this implies that $v_2 \geq a_2$. If $a_2 = a_3$, this follows from $v_2 \geq v_3$ and the fact that $a_2 + a_3 = v_2 + v_3$. Otherwise, if $a_2 > a_3$, we know that by the definition of the phantom rule and $a_2 = \min\{v_2,\Tilde{w}_2\}$.

                        \smallskip 

                        If $a_2 = v_2$, then $\ellone{a}{v} = 0 \leq \ellone{a^*}{v}$. Hence, we can assume that $a_2 = \max\{\frac{1-v_1}{2}, \min\{\Tilde{w}_2,v_2\}\} < v_2$ and as a direct consequence $a_3 > v_3$. 

                        \smallskip

                        Now, if $a^*_2<a_2$, then the conditions of \Cref{lem:twocoordinatesmove} are satisfied for alternatives $1$ and $2$, and hence $\ellone{a}{v} \leq \ellone{a^*}{v}$.

                        \smallskip 
                        Hence, assume that $a^*_2\geq a_2$. In the remainder of the proof, we show that $a^*_2 \leq a_2 + (v_1-v_1^*)/2$. By normalization (i.e., since $a^*_1 = a_1 - (v_1 - v^*_1)$) this implies that $a_3^* > a_3 > v_3$. Hence, the conditions for \Cref{lem:twocoordinatesmove} are satisfied for alternatives $1$ and $3$, and therefore $\ellone{a^*}{v}\geq \ellone{a}{v}$.

                        We define $\Tilde{w}^\star = \Tilde{c}(w,v^\star)$. To prove the claim, first note that, by definition of the phantom rule, it holds that $a^*_2 \leq \max\{(1-v^*_1)/2, \min\{\Tilde{w}^*_2,v^*_2\}\}$.
                        In the following we upper bound both of the terms within the maximum expression. For the first term, note that $$(1-v^*_1)/2 = (1-v_1)/2 +  (v_1 - v^*_1)/2.$$ 
                        For the second term, note that  
                        \begin{align*}
                        \gamma(w,v^\star) - \gamma(w,v) &= (w_1 - 0.8) \min \{\frac{0.8 - v^*_1}{0.1},1\} - (w_1 - 0.8) \min \{\frac{0.8 - v_1}{0.1},1\} \\
                        & \leq (w_1 - 0.8)\frac{(v_1 - v^*_1)}{0.1} \leq v_1 - v^*_1
                        \end{align*}
                        and therefore 
                        $$\Tilde{w}^*_2 = w_2 + \gamma(w,v^\star)/2 \leq w_2 + \gamma(w,v)/2 + (v_1 - v^*_1)/2 = \Tilde{w}_2 + (v_1-v_1^*)/2.$$ 

                        Therefore, we get that 
                        \begin{align*}
                        a^\star_2 &\leq \max\{(1-v^\star_1)/2, \min \{\Tilde{w}_2^\star,v_2^\star\}\} \\ 
                        &\leq \max\{(1-v^\star_1)/2, \Tilde{w}_2^\star\} \\ 
                        & \leq \max\{(1-v_1)/2 + (v_1-v_1^*)/2, \Tilde{w}_2 + (v_1-v_1^*)/2\} \\ 
                        & = \max\{(1-v_1)/2 , \Tilde{w}_2 \} + (v_1-v_1^*)/2 \\ 
                        & = a_2 + (v_1-v_1^*)/2, 
                        \end{align*}
                        which proves the claim and therefore the case. 
                    \end{CaseTree}
                \end{CaseTree}
            \end{CaseTree}
        \end{CaseTree}
    \end{CaseTree}
    \vspace*{-.9cm}
\end{proof}

\section{Missing Proofs of \Cref{sec:lowerBound}}

As stated in \Cref{sec:lowerBound}, the profile that we use to show \Cref{thm:truthful_uniform} is 
$P = (\p_1, \dots, \p_n)$ with 
\[
    \p_i = 
    \begin{cases}
        (\frac{1}{m}, \dots, \frac{1}{m}) & \text{ if } i \le \lceil\frac{n}{2}\rceil\\    
        (1, 0, \dots, 0) & \text{ if } i > \lceil\frac{n}{2}\rceil    
    \end{cases}
\]
and the proof makes use of the symmetry of this and similar profiles. More precisely, we consider profiles in which each vote lies on the line segment from $(1, 0, \dots, 0)$ to $(0, \frac{1}{m-1}, \dots, \frac{1}{m-1})$. For better readability, for a value $\alpha \in [0,1]$ we write the vote $(\alpha, \frac{1-\alpha}{m-1}, \dots, \frac{1-\alpha}{m-1})$ as $\symvote{\alpha}$. Then, any vote of such a profile can be written as $\symvote{\alpha_i}$ for some value $\alpha_i \in [0,1]$. Observe that for any values $x,y \in [0,1]$ we have $\ellone{p(x)}{p(y)} = 2|x-y|$. To simplify the notation for repeated votes in a profile, we write $\multivoter{\p}{k}$ for a series of $k$ voters with the same vote $\p$. We can then rewrite the profile $P$ as $P = (\multivoter{\symvote{\frac{1}{m}}}{\lceil\frac{n}{2}\rceil}, \multivoter{\symvote{1}}{\lfloor\frac{n}{2}\rfloor})$. Finally, it will be convenient to denote all values in between two values $x,y \in [0,1]$ as $\llbracket x,y \rrbracket = [x,y] \cup [y,x]$ and similarly, $\llbracket x,y \rrparenthesis = [x,y) \cup (y,x]$.

In preparation for proving \cref{thm:truthful_uniform}, we first make general observations on how CTAN mechanisms behave on the profiles described above. \Cref{lem:truthfulness_half_space_line:symmetry} states that, by neutrality, the aggregate must also be of the form $\symvote{\gamma}$ for some $\gamma \in [0,1]$. \Cref{lem:truthfulness_half_space_line:const_aggregate} considers situations in which voter 1 manipulates their vote so that the misreport $p^*_1$ still lies on the line segment, i.e., $p^*_1=p(\alpha_1')$ for some $\alpha_1' \neq \alpha_1$. We prove that if the voter does not `cross' the aggregate, then the aggregate cannot change due to truthfulness. Finally, \Cref{lem:truthfulness_half_space_line:min_max} shows that if the convex hull of the votes includes the center point of the simplex, then the aggregate must lie inside that convex hull as well. 

\begin{restatable}{lemma}{lemHalfSpace} \label{lem:truthfulness_half_space_line}
    \crefalias{enumi}{lemmaenumi}%
    \setlist[enumerate,1]{
        label={\textit{(\roman*)}},
        ref={\thelemma(\roman*)}
    }
    For any truthful, neutral, and continuous mechanism $\A$ and any $\alpha_1, \dots, \alpha_n \in [0,1]$ it holds that
    \begin{enumerate}
        \item $\A(\symvote{\alpha_1}, \dots, \symvote{\alpha_n}) = \symvote{\gamma}$, for some $\gamma \in [0,1]$.\label{lem:truthfulness_half_space_line:symmetry}
    \end{enumerate}
In the following, let $\gamma \in [0,1]$ be such that $\A(\symvote{\alpha_1}, \dots, \symvote{\alpha_n}) = \symvote{\gamma}$.
    \begin{enumerate}[resume]
        \item For any $i \in [n]$ and $\alpha_i^\star \in [0,1]$ with $\gamma \notin \llbracket \alpha_i, \alpha_i^\star\rrparenthesis$ it holds that $$\A(\symvote{\alpha_1}, \dots, \symvote{\alpha_{i-1}}, \symvote{\alpha_i^\star}, \symvote{\alpha_{i+1}},  \dots, \symvote{\alpha_n}) = \symvote{\gamma}.\label{lem:truthfulness_half_space_line:const_aggregate}$$
        \item If $\min_{i \in [n]} \alpha_i \le \frac{1}{m} \le \max_{i \in [n]} \alpha_i$, then $\min_{i \in [n]} \alpha_i \le \gamma \le \max_{i \in [n]} \alpha_i$.\label{lem:truthfulness_half_space_line:min_max}
    \end{enumerate}
\end{restatable}
\begin{proof}[Proofs]\leavevmode
    \begin{enumerate}[label=(\itshape\roman*)]
        \item The statement follows from neutrality. Since every vote of the form $p(\alpha_i)$ reports the same value for all alternatives in $\{2,\dots,m\}$, the aggregate of any neutral mechanism also has to assign one value to all alternatives in $\{2,\dots,m\}$. Any vote of this form can be expressed by $p(\gamma)$ with $\gamma \in [0,1]$.
        
        \item We prove the statement for the case when $\alpha_i < \gamma$. The other case can be shown analogously. We can additionally assume that $\alpha_i^\star\in [\alpha_i, \gamma]$, since otherwise we can swap the definitions of $\alpha_i$ and $\alpha_i^\star$. Let $P = (\symvote{\alpha_1}, \dots, \symvote{\alpha_n})$ be the original profile and $P^\star = (\symvote{\alpha_1}, \dots, \symvote{\alpha_{i-1}}, \symvote{\alpha_i^\star}, \symvote{\alpha_{i+1}},  \dots, \symvote{\alpha_n})$ the profile with the modified vote of voter $i$.
        
        Suppose $\A(P^\star) = \symvote{\gamma^\star} \neq \symvote{\gamma}$, where the existence of $\gamma^\star$ is guaranteed by \cref{lem:truthfulness_half_space_line:symmetry}.
        We distinguish three cases.
        \begin{CaseTree}
            \item $\gamma^\star>\gamma$. A voter with preference $p(\alpha_i^\star)$ prefers the outcome when misreporting $p(\alpha_i)$, contradicting the truthfulness of the mechanism:
            \begin{align*}
                \ellone{\symvote{\alpha_i^\star}}{\A(P)} 
                & = \ellone{\symvote{\alpha_i^\star}}{\symvote{\gamma}}
                = 2(\gamma  - \alpha_i^\star)  \\
                & < 2(\gamma^\star - \alpha_i^\star)
                = \ellone{\symvote{\alpha_i^\star}}{\symvote{\gamma^\star}}
                = \ellone{\symvote{\alpha_i^\star}}{\A(P^\star)}.
            \end{align*}
            \item $\alpha_i \leq \gamma^\star < \gamma$. A voter with preference $p(\alpha_i)$ prefers the outcome when misreporting $p(\alpha_i^\star)$, contradicting the truthfulness of the mechanism:
            \begin{align*}
                \ellone{\symvote{\alpha_i}}{\A(P^\star)}
                &= \ellone{\symvote{\alpha_i}}{\symvote{\gamma^\star}} 
                = 2(\gamma^\star - \alpha_i) \\ 
                &< 2(\gamma - \alpha_i) 
                = \ellone{\symvote{\alpha_i}}{\symvote{\gamma}} 
                = \ellone{\symvote{\alpha_i}}{\A(P)}
            \end{align*}
            \item $\gamma^\star < \alpha_i$. Then, we can use the continuity of the mechanism to derive a contradiction, as follows. Since $\gamma^\star<\alpha_i\leq \alpha_i^\star\leq \gamma$, by continuity there must be an $\alpha_i'$ with $\alpha_i \le \alpha_i' \le \alpha_i^\star$
            such that $\A\big(\symvote{\alpha_1}, \dots, \symvote{\alpha_{i-1}}, \symvote{\alpha_i'}, \symvote{\alpha_{i+1}},  \dots, \symvote{\alpha_n}\big) = p(\alpha_i')$. But then we can use the argument from Case 2, with $\alpha_i^\star = \alpha_i'$
            and $\gamma^\star=\alpha_i'$, to arrive at a contradiction.
        \end{CaseTree}

        \item Suppose $\gamma < \min_{i \in [n]} \alpha_i$ (the case $\gamma > \max_{i \in [n]} \alpha_i$ can be shown analogously). Then, for any $i \in [n]$, we can replace the vote $\symvote{\alpha_i}$ with $\symvote{\frac{1}{m}}$ and know by \Cref{lem:truthfulness_half_space_line:const_aggregate} that $\symvote{\gamma} = \A(\symvote{\alpha_1}, \dots, \symvote{\alpha_n}) = \A(\multivoter{\symvote{\frac{1}{m}}}{n})$. By neutrality, $\A(\multivoter{\symvote{\frac{1}{m}}}{n}) = \symvote{\frac{1}{m}}$, a contradiction.\qedhere
    \end{enumerate}
\end{proof}

We now give the proof of \Cref{thm:truthful_uniform}.

\maintheoremtwo*

\begin{proof}[Proof of \Cref{thm:truthful_uniform}]
    We prove the theorem using the profile $P = (\multivoter{\symvote{\frac{1}{m}}}{\lceil\frac{n}{2}\rceil}, \multivoter{\symvote{1}}{\lfloor\frac{n}{2}\rfloor})$ and showing that any CTAN mechanism $\A$ outputs $\A(P) = \symvote{\frac{1}{m}}$.

    Suppose there is a CTAN mechanism $\A$ that outputs $\A(P) \neq \symvote{\frac{1}{m}}$.
    By \Cref{lem:truthfulness_half_space_line:symmetry} we know that $\A(P)$ can be written as $\symvote{\beta}$ for some $\beta \in [0,1]$. Then, $\beta \neq \frac{1}{m}$ by assumption and together with Lemma \ref{lem:truthfulness_half_space_line:min_max} we know that $\frac{1}{m} < \beta \leq 1$.
    
    We introduce another profile $P' = (\multivoter{\symvote{\frac{1}{m}}}{\lfloor\frac{n}{2}\rfloor}, \multivoter{\symvote{0}}{\lceil\frac{n}{2}\rceil})$, for which we know by \Cref{lem:truthfulness_half_space_line:symmetry} and \ref{lem:truthfulness_half_space_line:min_max} that $\A(P') = \symvote{\gamma}$ for some $0 \le \gamma \le \frac{1}{m}$. \Cref{fig:lower_bound_two_profiles} shows the profiles $P$ and $P'$ for the case $m=3$. We consider two cases and show a contradiction in each.
    
    \begin{figure*}[t!]
        \begin{subfigure}[t]{0.475\textwidth}
            \begin{tikzpicture}[scale=5.5]
                \plotSimplex
                \small
                \node[draw, circle, minimum size=20, fill=white, fill opacity=1] at (0.5,0.286) (callout){};
                \draw[color1, fill] (0.5,0.286) node[voter, fill, xshift=-2.5, yshift=3.75](center1){};
                \draw[color2] (0.5,0.286) node[voter, fill=white, xshift=2.5, yshift=-3.75](center2){};
                \node[color2, label, anchor=north west] at (center2) {$\multivoter{\symvote{\frac{1}{m}}}{\frac{n}{2}}$};    
                \node[color1, label, anchor=south east, xshift=1pt] at (center1) {$\multivoter{\symvote{\frac{1}{m}}}{\frac{n}{2}}$};    
                \draw[color1, fill] (0,0) node[voter, fill]{} node[label, anchor=south east]{$\multivoter{\symvote{1}}{\frac{n}{2}}$};
                \draw[color2] (0.75,0.433) node[voter, fill=white]{} node[label, anchor=south west]{$\multivoter{\symvote{0}}{\frac{n}{2}}$};    
            
                \draw[color1] (0.35*0.75,0.35*0.433) node[aggregate, fill]{} node[label, anchor=south east, xshift=-0.5pt]{$\A(P) = \symvote{\beta}$};    
                \draw[color2] (0.85*0.75,0.85*0.433) node[aggregate, fill=white]{} node[label, anchor=north west, xshift=-1pt]{$\A(P') = \symvote{\gamma}$};    
                 
            \end{tikzpicture}
            \caption{Profiles $P$ (orange, filled) and $P'$ (blue, outline) with their corresponding aggregates. \Cref{thm:truthful_uniform} shows that $\beta = \frac{1}{m}$.}
            \label{fig:lower_bound_two_profiles}
        \end{subfigure}%
        \hfill
        \begin{subfigure}[t]{0.475\textwidth}
            \begin{tikzpicture}[scale=5.5]
                \plotSimplex
                \small
                \node[draw, circle, minimum size=20, fill=white, fill opacity=1] at (0.5,0.286) (callout){};
                
                \draw[color1, fill] (callout) node[voter, fill, xshift=-4.1, yshift=2.9](center1){};
                \draw[color2] (callout) node[voter, fill=white, xshift=4.1, yshift=2.9](center2){};
                \node[color2, label, anchor=south west, xshift=-1pt] at (center2) {$\multivoter{\symvote{\frac{1}{m}}}{\frac{n}{2}}$};   
                \node[color1, label, anchor=south east, yshift=0pt, xshift=2pt] at (center1) {$\multivoter{\symvote{\frac{1}{m}}}{\frac{n}{2}}$};    
                \draw[color1, fill] (0.25*0.75,0.25*0.433) node[voter, fill]{} node[label, anchor=north east, xshift=-1pt, yshift=2pt]{$\multivoter{\symvote{\alpha}}{\frac{n}{2}}$};    
                \draw[color2] (0.5,0) node[voter, fill=white]{} node[label, anchor=south west]{$\multivoter{\hat{p}}{\frac{n}{2}}$};
    
                \draw[color1] (0.47*0.75,0.47*0.433) node[aggregate, fill]{} node[label, anchor=east, yshift=1.8pt, xshift=-2pt]{$\A(P^\star) = \symvote{\beta'}$};  
                \draw[color2] (callout) node[aggregate, fill=white, xshift=0, yshift=-5](center3){} node[label, anchor=north west, xshift=-2pt, yshift=-4pt]{$\A(\hat{P}) = \symvote{\frac{1}{m}}$};  
    
                \draw[dashed, ->, thin, shorten <= 5pt, shorten >= 5pt] (0.5,0) --(0.25*0.75,0.25*0.433);
            \end{tikzpicture}
            \caption{Profiles $P^\star = (\multivoter{\symvote{\frac{1}{m}}}{\frac{n}{2}}, \multivoter{\symvote{\alpha}}{\frac{n}{2}})$ (orange, filled) and $\hat{P} = (\multivoter{\symvote{\frac{1}{m}}}{\frac{n}{2}}, \multivoter{\hat{p}}{\frac{n}{2}})$ (blue, outline) with their corresponding aggregates. The parametric profile $\hat{P}_k$ transforms $\hat{P}_0 = \hat{P}$ into $\hat{P}_{\frac{n}{2}} = P^\star$ by moving a voter from, from $\hat{p}$ to $\symvote{\alpha}$ whenever $k$ increases.}
            \label{fig:lower_bound_parametric_profile}
        \end{subfigure}
        \caption{Visualization of the profiles used in the proof of \Cref{thm:truthful_uniform} for $m = 3$ and even $n$. Circles denote voter groups and triangles denote aggregates. All voters and aggregates in the white circle are positioned at $p(\frac{1}{m})$. \cref{fig:lower_bound_two_profiles} shows the profiles from Case 1 of the proof and \cref{fig:lower_bound_parametric_profile} the profiles of Case 2.}
    \end{figure*}
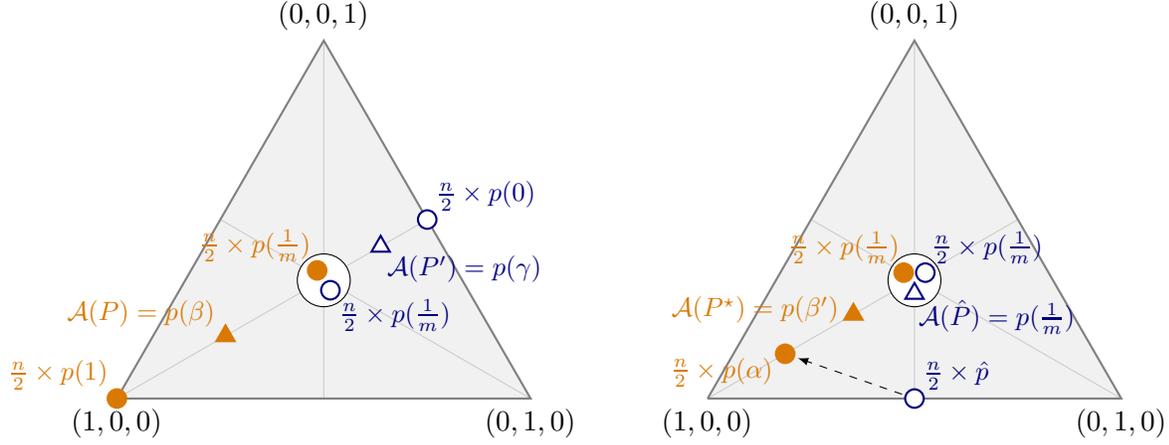

    \begin{CaseTree}
        \item It holds that $\gamma < \frac{1}{m}$.\\
        We make use of Lemma \ref{lem:truthfulness_half_space_line:const_aggregate} repeatedly to transform both $P$ and $P'$ into the same profile, without changing the output of the mechanism, which will give us a contradiction.
        
        Since $\A(P) = \A(\multivoter{\symvote{\frac{1}{m}}}{\lceil\frac{n}{2}\rceil}, \multivoter{\symvote{1}}{\lfloor\frac{n}{2}\rfloor}) = \symvote{\beta}$ and $\beta \notin \llbracket \frac{1}{m}, \gamma \rrparenthesis$, we can apply Lemma \ref{lem:truthfulness_half_space_line:const_aggregate} repeatedly to move all votes from $\symvote{\frac{1}{m}}$ to $\symvote{\gamma}$ without changing the aggregate, i.e. $\A(\multivoter{\symvote{\gamma}}{\lceil\frac{n}{2}\rceil}, \multivoter{\symvote{1}}{\lfloor\frac{n}{2}\rfloor}) = \symvote{\beta}$.
        We apply Lemma \ref{lem:truthfulness_half_space_line:const_aggregate} again to move all votes from $\symvote{1}$ to $\symvote{\beta}$ and get that $\A(\multivoter{\symvote{\gamma}}{\lceil\frac{n}{2}\rceil}, \multivoter{\symvote{\beta}}{\lfloor\frac{n}{2}\rfloor}) = \symvote{\beta}$ (since $\beta \notin \llbracket 1, \beta \rrparenthesis$).
    
        Similarly, for $P'$, we start from $\A(P') = \A(\multivoter{\symvote{\frac{1}{m}}}{\lfloor\frac{n}{2}\rfloor}, \multivoter{\symvote{0}}{\lceil\frac{n}{2}\rceil}) = \symvote{\gamma}$. We use Lemma \ref{lem:truthfulness_half_space_line:const_aggregate} to show that $\A(\multivoter{\symvote{\beta}}{\lfloor\frac{n}{2}\rfloor}, \multivoter{\symvote{0}}{\lceil\frac{n}{2}\rceil}) = \symvote{\gamma}$ (since $\gamma \notin \llbracket \frac{1}{m}, \beta \rrparenthesis$). Finally, knowing that $\gamma \notin \llbracket 0, \gamma \rrparenthesis$, we apply \Cref{lem:truthfulness_half_space_line:const_aggregate} again to get $$ \A(\multivoter{\symvote{\beta}}{\lfloor\tfrac{n}{2}\rfloor}, \multivoter{\symvote{\gamma}}{\lceil\tfrac{n}{2}\rceil}) \,\, =  \,\, \symvote{\gamma} \,\, \neq \,\, \symvote{\beta} \,\, = \,\, \A(\multivoter{\symvote{\gamma}}{\lceil\tfrac{n}{2}\rceil}, \multivoter{\symvote{\beta}}{\lfloor\tfrac{n}{2}\rfloor}),$$ contradicting anonymity. 

        \item It holds that $\gamma = \frac{1}{m}$.\\
        We now look at the profile $P''=(\multivoter{\symvote{\frac{1}{m}}}{\lceil\tfrac{n}{2}\rceil}, \multivoter{\symvote{0}}{\lfloor\tfrac{n}{2}\rfloor})$. If $n$ is even, this profile equals $P'$ and thus $A(P'') = A(P') = \symvote{\frac{1}{m}}$. If $n$ is odd, we can move one voter from $\symvote{0}$ to $\symvote{\frac{1}{m}}$ to construct $P''$ from $P'$ and get $A(P'') = A(P') = \symvote{\frac{1}{m}}$ by \cref{lem:truthfulness_half_space_line:const_aggregate}.

        By using continuity and the intermediate value theorem one can show that there exists some value $\alpha \in [\frac{1}{m},1]$ with the property that, for $P^\star = (\multivoter{\symvote{\frac{1}{m}}}{\lceil\tfrac{n}{2}\rceil}, \multivoter{\symvote{\alpha}}{\lfloor\tfrac{n}{2}\rfloor})$, it holds that $\A(P^\star) = \symvote{\beta'}$ with $\frac{1}{m} < \beta' < \frac{1}{m-1}$.\footnote{To see why, consider the function $f: [\frac{1}{m},1] \rightarrow [\frac{1}{m},1]$ defined as $f(x) = p^{-1}(\A(\multivoter{\symvote{\frac{1}{m}}}{\lceil\tfrac{n}{2}\rceil}, \multivoter{\symvote{x}}{\lfloor\tfrac{n}{2}\rfloor}))$, where the fact that $f$ is well-defined follows from \Cref{lem:truthfulness_half_space_line}. Moreover, $f(\frac{1}{m}) = \frac{1}{m}$ and $f(1) = \beta > \frac{1}{m}$ and the continuity of $\A$ implies that $f$ is continuous. Hence, by the intermediate value theorem we get that there exists $\alpha \in [\frac{1}{m},1]$ with $f(\alpha) \in [\frac{1}{m},\frac{1}{(m-1)}]$. Thus, $\alpha$ and $\beta' = f(\alpha)$ satisfy the statement of the claim.}

        We define the profile $\hat{P} = (\multivoter{\symvote{\frac{1}{m}}}{\lceil\tfrac{n}{2}\rceil}, \multivoter{\hat{p}}{\lfloor\tfrac{n}{2}\rfloor})$, where $\hat{p} = (\frac{1}{m-1}, \dots, \frac{1}{m-1}, 0)$. Note that this is, up to the order of the alternatives, the same profile as $P''$ and thus, by neutrality we know that $\A(\hat{P}) = \A(P'') = \symvote{\frac{1}{m}}$. We transform $\hat{P}$ into $P^\star$ in $\big\lfloor\frac{n}{2}\big\rfloor$ steps, while arguing how the output of $\A$ can change in each step. The profiles $P^\star$ and $\hat{P}$ are depicted in \Cref{fig:lower_bound_parametric_profile}.
    
        For $0 \le k \le \lfloor\tfrac{n}{2}\rfloor$, let $\hat{P}_k = (\multivoter{\symvote{\frac{1}{m}}}{\lceil\tfrac{n}{2}\rceil}, \multivoter{\hat{p}}{(\lfloor\tfrac{n}{2}\rfloor-k)}, \multivoter{\symvote{\alpha}}{k})$. Note that $\hat{P}_0 = \hat{P}$ and $\hat{P}_{\big\lfloor\tfrac{n}{2}\big\rfloor} = P^\star$. By truthfulness, we know for any $0 \le k \le \lfloor\tfrac{n}{2}\rfloor-1$ that $\ellone{\hat{p}}{\A(\hat{P}_k)} \le \ellone{\hat{p}}{\A(\hat{P}_{k+1})}$ and thus $\ellone{\hat{p}}{\A(\hat{P}_0)} \le \ellone{\hat{p}}{\A(\hat{P}_{\lfloor\tfrac{n}{2}\rfloor})}$. This is a contradiction, since we can compute the above $\ell_1$-distances as follows: \begin{align*}
            \ellone{\hat{p}}{\A(\hat{P}_{\big\lfloor\tfrac{n}{2}\big\rfloor})} 
            & = \ellone{\hat{p}}{\symvote{\beta'}} \\ 
            &= \left|\frac{1}{m-1} - \beta'\right| + \left|0 - \frac{1-\beta'}{m-1}\right| + (m-2) \left|\frac{1}{m-1} - \frac{1-\beta'}{m-1}\right| \\
            &= \frac{1}{m-1} - \beta' + \frac{1-\beta'}{m-1} + \frac{\beta' (m-2)}{m-1} = 2\frac{1 - \beta'}{m-1}\\
            &< 2\frac{1 - \frac{1}{m}}{m-1} = \frac{2}{m} = \left|0 - \frac{1}{m}\right| + (m-1) \left|\frac{1}{m-1} - \frac{1}{m}\right| \\
            &= \ellone{\hat{p}}{\symvote{\frac{1}{m}}} = \ellone{\hat{p}}{\A(\hat{P}_0)}.
        \end{align*}
    \end{CaseTree}   
    This concludes the prove of the second case. 

    We have shown that any CTAN mechanism outputs $\symvote{\frac{1}{m}}$ for the profile $P$. We now verify that this indeed yields the theorem statement. The \textsc{Mean} mechanism outputs $\Amean(P) = \symvote{\frac{m+1}{2m}}$ for even $n$ and $\Amean(P) = \symvote{\frac{m+1}{2m} - \frac{m-1}{2mn}}$ for odd $n$. This yields an $\ell_1$-distance of
    $$
    \ellone{\symvote{\frac{1}{m}}}{\mu(P)} =
    \begin{cases}
        2 \big|\frac{m+1}{2m} - \frac{1}{m}\big| = \frac{m-1}{m} & \text{ if $n$ is even}\\
        2 \big|\frac{m+1}{2m} - \frac{m-1}{2mn} - \frac{1}{m}\big| = \frac{m-1}{m}\cdot\frac{n-1}{n} & \text{ if $n$ is odd}
    \end{cases}
    $$
    and an $\ell_\infty$-distance of 
    $$
    \ellinf{\symvote{\frac{1}{m}}}{\mu(P)} =
    \begin{cases}
        \big|\frac{m+1}{2m} - \frac{1}{m}\big| = \frac{m-1}{2m} & \text{ if $n$ is even}\\
        \big|\frac{m+1}{2m} - \frac{m-1}{2mn} - \frac{1}{m}\big| = \frac{m-1}{2m}\cdot\frac{n-1}{n} & \text{ if $n$ is odd}.
    \end{cases}
    $$
    This concludes the proof. 
\end{proof}

\end{document}